\documentclass[reprint,amsmath,amssymb,aps,pra]{revtex4-2}
\usepackage[T1]{fontenc}

\makeatletter
\def\l@subsection#1#2{}
\def\l@subsubsection#1#2{}
\makeatother

\usepackage{graphicx}
\usepackage{dcolumn}
\usepackage{bm}
\usepackage{amsthm}
\usepackage{thmtools}
\usepackage{dsfont}
\usepackage{physics}
\newtheorem{definition}{Definition}
\newtheorem{proposition}{Proposition}

\newtheorem*{informaltheorem*}{Theorem (informal)}

\newtheorem{problem}{Problem}

\usepackage{amsfonts}
\usepackage{booktabs}
\usepackage{hyperref}
\usepackage{microtype}
\usepackage{xcolor}
\usepackage{color,soul}

\usepackage[normalem]{ulem}

\usepackage{algorithm}
\usepackage{algpseudocode}
\algtext*{EndWhile}
\algtext*{EndIf}

\newcommand{\calC}[0]{\mathcal{C}}
\newcommand{\calV}[0]{\mathcal{V}}
\newcommand{\calE}[0]{\mathcal{E}}
\newcommand{\calL}[0]{\mathcal{L}}
\newcommand{\calQ}[0]{\mathcal{Q}}

\newcommand{\calN}[0]{\mathcal{N}}
\newcommand{\calH}[0]{\mathcal{H}}
\newcommand{\NN}[0]{\mathbb{N}}
\newcommand{\RR}[0]{\mathbb{R}}
\newcommand{\CC}[0]{\mathbb{C}}

\newcommand{\poly}[0]{\operatorname{poly}}
\newcommand{\clique}[0]{\operatorname{Cl}}

\newcommand{\sng}[0]{\operatorname{sng}}

\newcommand{\prbm}[1]{\textsc{#1}}
\newcommand{\class}[1]{$\mathsf{#1}$}
\newcommand{\qmaone}[1]{$\mathsf{QMA}_1$}
\newcommand{\dqcone}[1]{$\mathsf{DQC}_1$}

\newcommand{\yes}[0]{\textsc{yes}}
\newcommand{\no}[0]{\textsc{no}}

\begin{document}

\preprint{APS/123-QED}

\title{Testing the presence of balanced and bipartite components in a sparse graph is \qmaone{}-hard}

\author{Massimiliano Incudini}
\email{massimiliano.incudini@univr.it}
\affiliation{University of Verona, Verona, Italy}%

\author{Casper Gyurik}
\email{casper.gyurik@pasqal.com}
\affiliation{Pasqal SaS, 7 rue L\'{e}onard de Vinci, 91300 Massy, France}%

\author{Riccardo Molteni}
\email{r.molteni@liacs.leidenuniv.nl}
\affiliation{applied Quantum algorithms (aQa), Leiden University, Leiden, The Netherlands}%

\author{Vedran Dunjko}
\email{v.dunjko@liacs.leidenuniv.nl}
\affiliation{applied Quantum algorithms (aQa), Leiden University, Leiden, The Netherlands}%

\date{\today}

\begin{abstract}
Determining whether an abstract simplicial complex, a discrete object often approximating a manifold, contains multi-dimensional holes is a task deeply connected to quantum mechanics and proven to be \qmaone{}-hard by Crichigno and Kohler. This task can be expressed in linear algebraic terms, equivalent to testing the non-triviality of the kernel of an operator known as the Combinatorial Laplacian. In this work, we explore the similarities between abstract simplicial complexes and signed or unsigned graphs, using them to map the spectral properties of the Combinatorial Laplacian to those of signed and unsigned graph Laplacians. We prove that our transformations preserve efficient sparse access to these Laplacian operators. Consequently, we show that key spectral properties, such as testing the presence of balanced components in signed graphs and the bipartite components in unsigned graphs, are \qmaone{}-hard. These properties play a paramount role in network science. The hardness of the bipartite test is relevant in quantum Hamiltonian complexity, as another example of testing properties related to the eigenspace of a stoquastic Hamiltonians are quantumly hard in the sparse input model for the graph.
\end{abstract}

\maketitle

\onecolumngrid

\tableofcontents

\clearpage

\section{Introduction}

Characterizing graph properties is a fundamental task in computer science. These properties include determining whether a graph is disconnected, meaning some vertices are unreachable from others, or bipartite, meaning its set of vertices can be divided into two groups such that every edge connects vertices from different groups. We can also characterize important properties of more complex combinatorial structures. For example, in \emph{signed} graphs, where edges can have positive or negative signs\textsuperscript{\footnote{Signed graphs should not be confused with directed graphs. Specifically, the Laplacian of a signed graph is symmetric and takes values in ${-1, 0, 1}$, whereas the Laplacian of a directed graph is not symmetric and takes values in ${0, 1}$.}}, a graph is called \emph{balanced} if its vertices can be partitioned into two groups such that all edges within each group are positive and those between groups are negative. In abstract simplicial complexes, which generalize graphs by allowing interactions among more than two elements called \emph{simplices}, an important property to test is the presence of homological holes, which is of great interest in algebraic topology and topological data analysis. Some of these properties can be described using spectral graph theory~\cite{chung1997spectral}, where graphs and their generalizations are associated with linear algebraic operators, allowing us to reformulate certain graph properties in terms of conditions on the spectra of these operators. Important examples include the graph Laplacian $\calL^G$ for graphs $G$, the signed graph Laplacian $\calL^{G_s}$ for signed graphs $G_s$, and the $p$-Combinatorial Hodge Laplacian $\calL^{\Gamma}_p$ for abstract simplicial complexes $\Gamma$ and integer $p \in \NN$.

A significant body of literature has focused on the capabilities of classical computation in characterizing input graphs as discussed above and on whether quantum computing can provide any advantages. Quantum Hamiltonian complexity~\cite{gharibian2015quantum} explores these tasks in relation to quantum complexity classes. In this context, the Hamiltonian corresponds to an embedding of the Laplacian of the combinatorial structure into a potentially larger operator acting on the Hilbert space of \(n\) qubits. 
These Hamiltonians must be efficiently encoded. This is achievable for the Combinatorial Hodge Laplacian of abstract simplicial complexes with $n$ vertices if we can test whether a given subset of vertices forms a simplex in the input complex in $\order{\poly n}$ time. An example of this is the clique complex of a graph $G$ with $n$ vertices, denoted as $\clique(G)$. Similarly, efficient access can be obtained for the Laplacian of sparse signed or unsigned graphs with $\order{2^n}$ vertices when the adjacency list can be computed by a classical circuit of size $\order{\poly n}$.

Recently, there has been substantial interest in quantum algorithms for topological data analysis and, more generally, in understanding the computational complexity of problems related to homology over simplicial complexes. A key finding in this area is that determining whether a clique complex has homological holes, previously known to be NP-hard~\cite{adamaszek2016complexity}, has now been proven to be \qmaone{}-hard~\cite{crichigno2022clique}. This result establishes an unexpected bridge between the mathematical properties of a particular combinatorial structure, the abstract simplicial complex, and quantum computing. This connection was further reinforced by \textcite{king2023promise}, who demonstrated that there exists a promise class of weighted simplicial complexes for which determining the homology is both \class{QMA} and \qmaone{}-hard. In this work, we extend these connections between combinatorial structures and quantum computational complexity by showing that the following problems are also \qmaone{}-hard.

\begin{informaltheorem*}
Determining if a sparse signed graph has a balanced connected component is \qmaone{}-hard. Furthermore, the task under the promise that its smallest eigenvalue is either zero or not less than inverse-polynomially small is contained in \class{QMA}.
\end{informaltheorem*}

We prove this theorem via a reduction from the clique homology problem. Note that the promise variant mentioned here is not necessarily \qmaone{}-hard.

\begin{informaltheorem*}
Determining if a sparse unsigned graph has a bipartite connected component is \qmaone{}-hard. Furthermore, the task under the promise that its smallest eigenvalue is either zero or not less than inverse-polynomially small is contained in \class{QMA}.
\end{informaltheorem*}

We prove this theorem via a reduction from the task of determining the existence of balanced connected components. A recently discovered reduction from balance to bipartite testing is known in the literature \cite{zaslavsky2018negative}. We build upon this work by ensuring that the reduction preserves efficiently implementable sparse access to the graph.

The connection between the Combinatorial Hodge Laplacian and the Laplacian of the signed graph unveils a link between topological data analysis and network science, possibly leading this latter as a novel applicative area for quantum computing. In fact, signed graphs are used in network science to model entities (vertices) with relationships (edges) that can be either synergistic (positive) or adversarial (negative). Balanced graphs have been extensively studied for their ability to model specific behaviors within these networks, finding application, for example, in finance, marketing, biology, and ecology~\cite{figueiredo2014maximum, tang2016recommendations, tian2021extracting, vicsek1995novel, ou2015detecting, reynolds1987flocks}. Similarly, the unsigned graph is the most basic model used in network science; here, a bipartite network often exhibits hidden patterns and relationships (e.g., student-teacher, employer-employee).

Our work can be framed within the context of quantum Hamiltonian complexity. These network-related tasks represent a restriction of the sparse Hamiltonian problem. The focus on graphs with \emph{efficiently implementable} sparse access and the emphasis on the Laplacian operator distinguishes our work from the existing literature.

Graphs with \emph{efficiently implementable} sparse access typically exhibit simple or highly regular topologies. However, this framework can also implement notable families of graphs, such as graphs induced by clique complexes (as done in this work), balanced binary trees \cite{camps2024explicit}, Toeplitz graphs (whose adjacency matrix is a Toeplitz matrix \cite{sunderhauf2024block}), and circulant graphs, among others. The extent to which real-world networks commonly analyzed in network science, such as scale-free networks \cite{barabasi1999emergence}, can be mapped to these topologies, and whether it is possible to efficiently compress such graphs to fit these succinct descriptions, will be the subject of future investigation. A possible path forward for adapting graphs to such succinct access is the use of sparsification techniques, which have been investigated in \cite{herbert2019spectral, apers2022quantum}.

The literature on sparse Hamiltonians is extensive. Thanks to the connections between signed and unsigned graphs, this framework allows us to naturally discuss the Laplacian of unsigned graphs. The latter operator is \emph{stoquastic}. Notably, certain tasks that are \class{QMA}-hard for general Hamiltonians become easier when restricted to stoquastic Hamiltonians, such as ground state energy estimation. However, determining the excited states of stoquastic operators can still be \class{QMA}-complete~\cite{jordan2010quantum}. Here, we prove that testing for bipartite components, which depends non-trivially on the spectral properties of the Laplacian other than the ground state energy, is \qmaone{}-hard; this provides another example where estimating spectral properties of stoquastic Hamiltonians is quantumly hard. 

\subsection{Overview of the results}

The first part of this work unveils a connection between abstract simplicial complexes and signed graphs. 

\begin{restatable}{definition}{restatableasc}
An \emph{abstract simplicial complex} (ASC) is an ordered pair $\Gamma = (V, \Sigma)$ where $V$ is a non-empty set of vertices, and $\Sigma$ is a non-empty subset of $2^V \setminus \{ \varnothing \}$ that is closed under inclusion. The elements of $\Sigma$ are called simplices. The elements of $\Sigma_p \subset \Sigma$ are $p$-simplices, each of which has size $p+1$. A face of a $p$-simplex $\sigma$ is any $(p-1)$-simplex obtained by removing a single vertex from $\sigma$. Two $p$-simplices $\sigma, \tau$ are lower-adjacent, denoted as $\sigma \sim_l \tau$, if they share the face $\sigma \cap \tau \in \Sigma_{p-1}$. Two $p$-simplices $\sigma, \tau$ are upper-adjacent, denoted as $\sigma \sim_u \tau$, if are faces of a common $(p+1)$-simplex $\rho$.
\end{restatable}

A key relation between $p$-simplices in $\Gamma$ is captured by the operator known as the Combinatorial Hodge Laplacian, $\calL_p^\Gamma$. The elements of $\ker \calL_p^\Gamma$ can be associated with features known as $p$-dimensional holes in algebraic topology. We generalize the work in \textcite{jost2023cheeger}, showing the construction of a graph $G_s(\Gamma, p)$ such that $\calL_p^\Gamma$ and $\calL^{G_s}$ are unitarily equivalent.

\begin{restatable}{proposition}{restatablepropascsigned}\label{prop:full_laplacian} 
For any abstract simplicial complex $\Gamma = (V, \Sigma)$ and any $p \in \NN$, let $G_s = (V_s, E_s, s)$ be the signed graph having
\begin{align*} 
V_s & = \Sigma_p, \\ 
E_s & = \{ \{\sigma, \tau\} \mid \sigma, \tau \in \Sigma_p, \text{either } \sigma \sim_u \tau \text{ or } \sigma \sim_l \tau \}, \\ 
s(\sigma, \tau) & = - \Big(\sng(\sigma \cap \tau, \sigma) \sng(\sigma \cap \tau, \tau) + \sng(\sigma, \sigma \cup \tau) \sng(\tau, \sigma \cup \tau)\Big). 
\end{align*}
Then, $\calL_p^\Gamma$ and $\calL^{G_s}$ are unitarily equivalent. Here, $\sng$ is defined as per Definition \ref{def:sng}.
\end{restatable}

An important aspect to consider here is the input model. Some families of abstract simplicial complexes, such as the clique complex, are large combinatorial structures that admit a compact description of $\order{\poly n}$ bits, where $n$ is the number of vertices. Here, we consider sparse graphs with $N < 2^n$ vertices that allow for a similarly compact description. These graphs are provided through a classical circuit that enables sparse access to the graph in $\order{\poly n}$ time steps.

\begin{restatable}[Marked sparse access for signed graphs]{definition}{restatabledefsparsesigned}\label{def:bounded_degree_signed}
Let $G_s = (V, E, s)$ be a signed graph having $N \le 2^n - 1$ vertices, $V \subseteq [2^n] \setminus \{0\}$, and where each vertex has at most $S \in \order{\poly n}, S \ge 2$ neighbors. Let $\operatorname{adj}: [2^n] \times [S] \to [2^n]$ be the mapping such that
\begin{equation}
    \operatorname{adj}(i, \ell) = \begin{cases}
        j, & i \in V \text{ and } j \text{ is the } \ell\text{-th neighbor of } i \\
        0, & i \in V \text{ and } \ell \ge \deg(i)\\
        \ell, & i \not\in V \\
    \end{cases}.
\end{equation}
A sparse access for $G_s$ is given by the pair of oracles $O_\text{adj}, O_\text{sign}$,
\begin{align}
    O_\text{adj} \ket{i} \ket{\ell} & = \ket{\operatorname{adj}(i, \ell)} \ket{\ell}, \nonumber \\
    O_\text{sign} \ket{i} \ket{\ell} \ket{z} & = \ket{i} \ket{\ell} \ket{z \oplus s(\{i, \operatorname{adj}(i, \ell)\})}. \nonumber
\end{align}
\end{restatable}

Here, the bitstring $0 \not\in V$ serves as a placeholder indicating the end of the adjacency list. The sparse access returns a list of all zeros for any isolated vertex $i \in V$, and a marked list $[0, ..., S-1]$ for any bitstring $i \not\in V$. 

Note that this approach provides a natural way to define oracle access to the adjacency list of the graph. The only difference between our definition and that of \textcite{goldreich1997property} is that we explicitly encode the set of vertices as a set of \( n \)-bit natural numbers, \( V \subseteq [2^n] \setminus \{0\} \), with \( V \) not necessarily equal to \( \{1, \ldots, N\} \). This distinction plays an important role in proving our results. However, also we demonstrate that our results hold under the traditional input model for sparse graphs. Notably, for other tasks, the equivalence between these two models may not hold.

We prove that the task of testing the presence of balanced connected components of a signed graph in the marked sparse access is at least as hard as testing homologies in a clique complex. A balanced component contains no cycles with an odd number of negative edges. This is obtained via a reduction from the \prbm{homology} problem in \textcite{crichigno2022clique}. 

\begin{restatable}[\prbm{sparse balancedness}]{problem}{restatablesparsebalancedness} \phantom{.} \\
\emph{Input}: A signed graph $G_s$ with $N \le 2^n - 1$ vertices such that
\begin{itemize}
    \setlength\itemsep{0em}
    \item the graph is sparse, i.e., there is an upper bound $S \in \order{\poly n}$ on the degree of the vertices;
    \item the graph is given as a classical circuit of size $\order{\poly n}$ implementing the marked sparse access $(O_\text{adj}, O_\text{sign})$.
\end{itemize} 
\emph{Output}: \yes{} if $G_s$ has a balanced component, \no{} otherwise.
\end{restatable}

\begin{restatable}{theorem}{restatabletheoremone}\label{thm:thm_1}
\prbm{sparse balancedness} is \qmaone{}-hard.
\end{restatable}

Notably, a \emph{promise} variant of \prbm{sparse balancedness} is contained in \class{QMA}.

\begin{restatable}[\prbm{promise sparse balancedness}]{problem}{restatablepromisesparsebalancedness} \phantom{.} \\
\emph{Input}: A signed graph $G_s$ with $N \le 2^n - 1$ vertices such that
\begin{itemize}
    \setlength\itemsep{0em}
    \item the graph is sparse, i.e., there is an upper bound $S \in \order{\poly n}$ on the degree of the vertices;
    \item the graph is given as a classical circuit of size $\order{\poly n}$ implementing the marked sparse access $(O_\text{adj}, O_\text{sign})$.
\end{itemize} 
\emph{Promise}: the smallest eigenvalue of $\calL^{G_s}$ is either $0$ or $\ge \delta \in \Omega(1/\poly n)$ \\[0.5em]
\emph{Output}: \yes{} if $G_s$ has a balanced component, \no{} otherwise.
\end{restatable}

\begin{restatable}{proposition}{restatabletheoremoneqma}\label{prop:thm_1_qma}
\prbm{promise sparse balancedness} is contained in \class{QMA}.
\end{restatable}

The second part of this work studies a construction proposed by \textcite{zaslavsky2018negative} that connects signed and unsigned graphs, specifically, with respect to the properties of balancedness and bipartiteness. For any signed graph $G_s = (V, E, s)$, there exists an unsigned graph $G' = (V', E')$ of comparable size such that $G_s$ has a balanced component if and only if $G'$ has a bipartite component. A bipartite component contains no cycles of odd length.
We prove that when the input is restricted to signed graphs with sparse access, these guarantees hold for the output of the construction as well. A corollary of these results is the following statement regarding bipartiteness.

\begin{restatable}[\prbm{sparse bipartitedness}]{problem}{restatablesparsebipartitedness} \phantom{.} \\
\emph{Input}: An unsigned graph $G$ with $N \le 2^n - 1$ vertices such that
\begin{itemize}
    \setlength\itemsep{0em}
    \item the graph is sparse, i.e., there is an upper bound $S \in \order{\poly n}$ on the degree of the vertices;
    \item the graph is given as a classical circuit of size $\order{\poly n}$ implementing the marked sparse access $O_\text{adj}$.
\end{itemize} 
\emph{Output}: \yes{} if $G$ has a bipartite component, \no{} otherwise.
\end{restatable}

\begin{restatable}{theorem}{restatabletheoremtwo}\label{thm:thm_2}
\prbm{sparse bipartitedness} is \qmaone{}-hard.
\end{restatable}

Similarly to the previous case, we can identify a promise variant of \prbm{sparse bipartitedness} that is contained in \class{QMA}. It is not straightforward to study the presence of bipartite components in relation to the spectral properties of the graph Laplacian operator. \textcite{trevisan2017lecture} shows the largest eigenvalue of a certain normalized variant of the graph Laplacian is exactly $2$ if the graph has a bipartite component. However, another characterization exists stated in terms of the signless Laplacian operator $\calQ^G$, a modification of the graph Laplacian. Specifically, the dimensionality of the kernel of the signless Laplacian of an unsigned graph corresponds to the number of its connected bipartite components \cite{cvetkovic2007signless}. Consequently, checking the non-triviality of this kernel serves as a test for the presence of bipartite components.

\begin{restatable}[\prbm{promise sparse bipartitedness}]{problem}{restatablepromisesparsebipartitedness} \phantom{.} \\
\emph{Input}: An unsigned graph $G$ with $N \le 2^n - 1$ vertices such that
\begin{itemize}
    \setlength\itemsep{0em}
    \item the graph is sparse, i.e., there is an upper bound $S \in \order{\poly n}$ on the degree of the vertices;
    \item the graph is given as a classical circuit of size $\order{\poly n}$ implementing the marked sparse access $O_\text{adj}$.
\end{itemize} 
\emph{Promise}: the smallest eigenvalue of $\calQ^{G}$ is either $0$ or $\ge \delta \in \Omega(1/\poly n)$ \\[0.5em]
\emph{Output}: \yes{} if $G$ has a bipartite component, \no{} otherwise.
\end{restatable}

\begin{restatable}{proposition}{restatabletheoremtwoqma}\label{thm:thm_2_promise}
\prbm{promise sparse bipartitedness} is contained in \class{QMA}.
\end{restatable}

\subsection{Related works}

\paragraph{Hardness of Clique Homology}

Our work is grounded in recent advancements in quantum algorithms for topological data analysis. The pioneering study by \textcite{lloyd2016quantum} introduced the first quantum algorithm for estimating a quantity related to the topological invariant known as Betti numbers. This foundational work has been followed by several improvements~\cite{mcardle2022streamlined, berry2024analyzing}. The hardness of estimating Betti numbers, and related proxy measures, has been explored in~\cite{gyurik2022towards, schmidhuber2023complexity}. Notably, \textcite{cade2024complexity} demonstrated that estimating the Betti numbers of an abstract simplicial complex is \qmaone{}-hard. This result was subsequently refined by \textcite{crichigno2022clique}, who restricted the input to clique complexes, which can be succinctly described in polynomial size in to the number of vertices. Building on this, \textcite{king2023promise} proved that the promise variant of the clique homology problem is both \class{QMA} and \qmaone{}-hard when weighted cliques are allowed. \textcite{gyurik2024quantum} have proven that a variant of computing persistence is \class{BQP_1}-hard and in \class{BQP}.

\paragraph{Stoquastic Hamiltonians}

The Laplacian of an unsigned graph is an example of a stoquastic Hamiltonian. Stoquasticity occurs when the operator admits a matrix representation with all its off-diagonal elements non-positive. This allows us to characterize the ground state via the Perron-Frobenius theorem. Consequently, certain tasks over stoquastic Hamiltonians are easier compared to non-stoquastic ones. Examples include the Hamiltonian simulation and determination of the ground state energy, which \textcite{bravyi2006complexity, bravyi2006merlin} have shown to lie between the classes \class{MA} and \class{QMA}. These form a new class denoted as \class{StoqMA}. The local stoquastic Hamiltonian problem, i.e., determining whether the ground state energy is below a constant \( a \) or above a constant \( b = a + \Omega(1/\poly n) \), is \class{StoqMA}-complete. Further restrictions, such as imposing \( a = 0 \) and \( b = \Omega(1) \), belong to \class{NP}~\cite{aharonov2019stoquastic}. \textcite{jordan2010quantum} has shown that estimating the energy of excited states for stoquastic Hamiltonians is significantly harder than estimating its ground state, in particular, it is \class{QMA}-complete. We remark that testing for the presence of bipartite components in an unsigned graph does not appear to be derivable from the kernel of the Laplacian. Instead, a possible spectral characterization of bipartite components can be obtained by checking whether the largest eigenvalue of a certain normalized variant of the graph Laplacian is exactly 
$2$ (cf. \textcite[Theorem 3.2, Section 4.1]{trevisan2017lecture}). The presence of bipartite component can be also stated as testing the non-triviality of $\ker\calQ^G$, where $\calQ^G$ is the signless Laplacian operator.

\paragraph{Comparison with \textcite{childs2014bose}}

Our work is related to the results in \textcite[Appendix A]{childs2014bose}. In this study, the authors demonstrate that the task of estimating the smallest eigenvalue of an adjacency matrix (symmetric with coefficients in \(\{0, 1\}\)) is \class{QMA}-complete. This result holds if the underlying graph of the adjacency matrix admits a sparse access, and such can be efficiently implemented. In contrast, our work focuses specifically on the Laplacian operator. The two operators are connected non-trivially; although a property stated in terms of the adjacency matrix can often be translated to the Laplacian and vice versa, the connection is not always straightforward. Notably, the low-energy eigenspace of the adjacency matrix intuitively corresponds to the high-energy eigenspace of the Laplacian: this connection is clear only for \(d\)-regular graphs (where each vertex has \(d\) neighbors) for which \(\calL^G = d \mathbb{I} - \mathcal{A}\), where $\mathcal{D} = d\mathbb{I}$ is the diagonal degree matrix, and thus the eigenvalue \(\lambda\) of the Laplacian and \(\mu\) of the adjacency matrix are related by \(\lambda = d - \mu\). In general, this does not hold. A possible connection with our work may arise from the use of the signless Laplacian for the bipartiteness testing. That is because $\calQ^G$ can be stated as $\calQ^G = \mathcal{D} + \mathcal{A}$. This paves an alternative path of proving Theorem \ref{thm:thm_2} starting from the results in \textcite{childs2014bose}. However, such a path is non-trivial as relating the eigenvalues of $\calQ^G$ and $\mathcal{A}$ is challenging as is the case for $\calL^G$ and $\mathcal{A}$.

\subsection{Organization of the paper}

The remainder of this paper is devoted to proving Theorems~\ref{thm:thm_1} and \ref{thm:thm_2}.

In Section~\ref{sec:definition}, we review some fundamental concepts, with a particular emphasis on the input model for sparse graphs used throughout the document.

In Section~\ref{sec:balance}, we prove that the \prbm{sparse balancedness} task is \qmaone{}-hard by characterizing the homology of a clique complex in terms of the balanced components of a sparse signed graph. This leads to a reduction from the \prbm{clique homology} problem to \prbm{sparse balancedness}. We prove a promise variant of \prbm{sparse balancedness} is contained in \class{QMA}.

In Section~\ref{sec:bipartite}, we prove that the \prbm{sparse bipartitedness} task is \qmaone{}-hard by using a characterization of the balanced components of a signed graph in terms of the bipartition of an unsigned graph of comparable size. This results in a reduction from the \prbm{sparse balancedness} problem to \prbm{sparse bipartitedness}. We prove a promise variant of \prbm{sparse bipartitedness} is contained in \class{QMA}.

\subsection{Future directions}

\begin{enumerate}
    \item \emph{Do the \prbm{sparse balancedness} and \prbm{sparse bipartitedness} problems remain \qmaone{}-hard even for connected graphs?} \\
    A consequence of the reduction from \prbm{clique homology} to \prbm{sparse balancedness} is that the signed graph, which is constructed via Proposition \ref{prop:full_laplacian}, can have many connected components. Showing that \prbm{sparse balancedness} is \qmaone{}-hard for connected graphs would provide a more natural formulation of the problem, specifically, asking if the entire graph is balanced rather than merely having a balanced component. Such a result would immediately extend to \prbm{sparse bipartitedness} being \qmaone{}-hard for connected graphs and could potentially lead to novel results in characterizing properties of sparse unsigned graphs. 

    \item \emph{Is it equivalent to express a graph in the marked sparse access model and in the traditional sparse access model for \prbm{sparse balancedness} and \prbm{sparse bipartitedness} restricted to connected graphs?} \\
    The procedure that converts a graph from the marked sparse access model to the traditional sparse access model, a key component in proving that \prbm{sparse balancedness} and \prbm{sparse bipartitedness} have the same complexity under both input models, does not preserve the number of connected components. If we restrict \prbm{sparse balancedness} and \prbm{sparse bipartitedness} to connected graphs, is having the graph in the traditional sparse access model equivalent to having it in the marked sparse access model, or is the latter model, in some sense, more powerful?
    
    \item \emph{Is there a promise variant of \prbm{sparse balancedness} and \prbm{sparse bipartitedness} that is both \qmaone{}-hard and \class{QMA}?} \\
    A possible direction could involve \prbm{gapped clique homology}, which has been shown to be both \qmaone{}-hard and in \class{QMA} only for \emph{weighted cliques} in \textcite{king2023promise}, while the \emph{unweighted} case has yet to be proven. If such a result were established, we could use our construction to immediately obtain promise variants of \prbm{sparse balancedness} and \prbm{sparse bipartiteness}.

    \item \emph{Can we approximate a graph given as adjacency list with one in the marked sparse access while preserving certain properties such as it being balanced?} \\
    We focus on families of graphs that can be succinctly described by circuits capable of generating these representations. An intriguing open problem is determining whether we can extend the class of graphs that can be effectively analyzed in this way. One potential approach might involve graph sparsification techniques, as discussed in \cite{herbert2019spectral, apers2022quantum}. These techniques could help preserve properties such as balancedness and bipartiteness, while converting graphs to an efficiently implemented marked sparse access.
\end{enumerate}

\section{Definitions}\label{sec:definition}

\subsection{\qmaone{} complexity class}
We begin by recalling the definitions of the \class{QMA} and \qmaone{} decision classes.

\begin{definition}{(\class{QMA})}
Let $A = (A_\text{yes}, A_\text{no})$ be a promise problem. Then, $A \in$ \class{QMA} if there exists a polynomial time quantum verifier $V$ and polynomial $p$ such that for every input $x \in \{0, 1\}^n$:
\begin{itemize} 
    \item If $x \in A_\text{yes}$, there exists a witness state $\ket{w}$ over $p(n)$ qubits such that $V(x, \ket{w})$ accepts with probability $\ge 2/3$.
    \item If $x \in A_\text{no}$, for any witness state $\ket{w}$ over $p(n)$ qubits we have that $V(x, \ket{w})$ accepts with probability $\le 1/3$.
\end{itemize} 
\end{definition}

\begin{definition}{(\qmaone{})}
Let $A = (A_\text{yes}, A_\text{no})$ be a promise problem. Then, $A \in$ \class{QMA} if there exists a polynomial time quantum verifier $V_x$ composed of gates from a fixed universal set and polynomial $p$ such that for every input $x \in \{0, 1\}^n$:
\begin{itemize} 
    \item If $x \in A_\text{yes}$, there exists a witness state $\ket{w}$ over $p(n)$ qubits such that $V(x, \ket{w})$ accepts with probability $1$.
    \item If $x \in A_\text{no}$, for any witness state $\ket{w}$ over $p(n)$ qubits we have that $V(x, \ket{w})$ accepts with probability $\le 1/3$.
\end{itemize} 
\end{definition}

Problems that are \class{QMA}-hard and \qmaone{}-hard are believed to be intractable for both classical and quantum computers. However, some modified variants of these tasks have been found to be tractable on a quantum computer while remaining intractable on classical ones \cite{cade_et_al}.

The key distinction between \class{QMA} and \qmaone{} lies in the \emph{completeness} condition, i.e. the probability that the verifier accepts a valid input. Specifically, \qmaone{} is characterized by a \emph{perfect completeness}; if the input is valid, the verifier always accepts. This aspect prevents us from being independent with respect to the universal gate set used to define the quantum verifier: in general, we cannot define a quantum circuit equivalent to a given one, specified in a different universal gate set, with exactly zero error.
In our work, we root our results in a reduction from the \prbm{clique homology} problem defined in \textcite{cade2024complexity}. Consequently, we are compelled to use the same gate set, consisting of the \textsc{cnot} gate and the gate $U_\text{pyth} = \frac{1}{5} \smqty(3 & 4 \\ -4 & 3)$. This choice, however, has no direct impact on the validity or implications of our results.

\subsection{Basic definitions from graph theory}

We recall the definitions of a graph and a signed graph.

\begin{definition}{(Simple graph)}
A simple graph is an ordered pair $G = (V, E)$ where $V$ is a non-empty, finite set of vertices, and $E \subseteq \binom{V}{2}$ is a (possibly empty) set of edges. The set $\binom{V}{2}$ refers to the collection of all subsets of size 2 that can be formed from the elements of $V$. 
\end{definition}

Throughout this document, we assume that graphs are simple, undirected, and unweighted unless otherwise stated. The simplicity condition means that self-loops and multiple edges between the same pair of vertices are not allowed. Graphs can represented as a linear operator over the finite-dimensional Hilbert space spanned by the set of vertices of $G$, $\calV^G = \operatorname{span}(V, \RR)$, known as the graph Laplacian operator.

\begin{definition}\label{def:unsigned_laplacian}
The \emph{graph Laplacian} of a graph $G = (V, E)$ is a linear operator $\calL^G: \calV^G \to \calV^G$ defined as 
\begin{equation}
    \mel{i}{\calL^G}{j} = \begin{cases}
        \deg(i), & i = j \\
        -1, & i \sim j \\
        0, & \text{otherwise}
    \end{cases}.
\end{equation}
Here, $i \sim j$ indicates that the vertices $i$ and $j$ are adjacent, i.e., $\{i, j\} \in E$, while $\deg(i)$ is the degree of vertex $i$, corresponding to the number of its adjacent vertices, i.e., $\deg(i) = \sum_{i \sim j} 1$.
\end{definition}

We omit the explicit reference to $G$ when it is clear from the context. The graph Laplacian is sometimes expressed as $\calL = \mathcal{D} - \mathcal{A}$, where $\mathcal{D}$ is the (diagonal) degree operator, and $\mathcal{A}$ the adjacency matrix, which is the symmetric matrix with elements in $\{0, 1\}$. A signed graph is a generalization of the simple graph defined above. In this case, each edge is assigned a positive or negative sign. Formally,

\begin{definition}{(Signed graph)}
A \emph{signed graph} is an ordered pair $G_s = (V, E, s)$ where $V$ is a non-empty, finite set of vertices, $E \subseteq \binom{V}{2}$ is a (possibly empty) set of edges, and $s$ is the signature, a mapping $s: E \to \{-1, 1\}$. The set $\binom{V}{2}$ refers to the collection of all subsets of size 2 that can be formed from the elements of $V$. 
\end{definition}

\begin{restatable}{definition}{restatabledefsignedlaplacian}\label{def:signed_laplacian}
The \emph{signed Laplacian} is the operator $\calL: \calV \to \calV$ defined as 
\begin{equation}
    \mel{i}{\calL}{j} = \begin{cases}
        \deg(i), & i = j \\
        -s(i, j), & i \sim j \\
        0, & \text{otherwise}
    \end{cases}.
\end{equation}
Here, $i \sim j$ indicates that the vertices $i$ and $j$ are adjacent, i.e., $\{i, j\} \in E$, while $\deg(i)$ is the degree of vertex $i$, which corresponds to the number of its adjacent vertices, irrespective of the sign of the edge, i.e., $\deg(i) = \sum_{i \sim j} 1$.
\end{restatable}

\subsection{Marked sparse access for sparse graphs}

We introduce the graph input model that is used throughout this work.

\begin{restatable}[Marked sparse access for unsigned graphs]{definition}{restatabledefsparseunsigned}\label{def:sparse_access_unsigned}
Let $G = (V, E)$ be an unsigned graph having $N \le 2^n-1$ vertices, $V \subseteq [2^n] \setminus \{0\}$, and where each vertex has at most $S \in \order{\poly n}, S \ge 2$ neighbors. Let $\operatorname{adj}: [2^n] \times [S] \to [2^n]$ be the mapping such that
\begin{equation}
    \operatorname{adj}(i, \ell) = \begin{cases}
        j, & i \in V \text{ and } j \text{ is the } \ell\text{-th neighbor of } i \\
        0, & i \in V \text{ and } \ell \ge \deg(i) \\
        \ell, & i \not\in V 
    \end{cases}.
\end{equation}
A sparse access for $G$ is given by the oracle
\begin{align}
    O_\text{adj} \ket{i} \ket{\ell} & = \ket{\operatorname{adj}(i, \ell)} \ket{\ell}. \nonumber 
\end{align}
\end{restatable}

\restatabledefsparsesigned* 

We compare our model with the traditional model, as introduced, for example, by \textcite{goldreich1997property}. In that model, the authors define a function \( f_G: V \times [S] \to V \cup \{ \mathbf{0} \} \), where \( G = (V, E) \), \( V = \{ 1, \ldots, N \} \) is the set of vertices of an undirected graph, \( S \) is the upper bound on the degree of the vertices, and \( f_G(v, \ell) \) returns the \(\ell\)-th neighbor of \( v \in V \) if it exists, or \(\mathbf{0}\) otherwise. Since we need to implement this model using circuits, we assign a specific bitstring to the placeholder value; here, \(\mathbf{0}\) is represented by the bitstring \(0\). Similarly, the domain of \( f_G \) is limited to the set of vertices, so the behavior of the circuit for bitstrings \( i > N \) is undefined; in such cases, we define it to return the list \([0, 1, \ldots, S-1]\). The extension to signed graphs is straightforward.

\begin{restatable}[Traditional sparse access for unsigned graphs]{definition}{restatabledefsparseunsignednm}\label{def:sparse_access_unsigned_nm}
Let $G = (V, E)$ be an unsigned graph with $V = \{1, \ldots, N\}$, and where each vertex has at most $S \in \order{\poly n}$ neighbors. Let $\operatorname{adj}: [2^n] \times [S] \to [2^n]$ be the mapping such that
\begin{equation}
    \operatorname{adj}(i, \ell) = \begin{cases}
        j, & 1 \le i \le N \text{ and } j \text{ is the } \ell\text{-th neighbor of } i \\
        0, & 1 \le i \le N \text{ and } \ell \ge \deg(i) \\
        \ell, & \text{otherwise} 
    \end{cases}.
\end{equation}
A sparse access for $G$ is given by the oracle
\begin{align}
    O_\text{adj} \ket{i} \ket{\ell} & = \ket{\operatorname{adj}(i, \ell)} \ket{\ell}. \nonumber 
\end{align}
\end{restatable}

\begin{restatable}[Traditional sparse access for signed graphs]{definition}{restatabledefsparsesignednm}\label{def:sparse_access_signed_nm}
Let $G_s = (V, E, s)$ be a signed graph with $V = \{1, \ldots, N\}$, and where each vertex has at most $S \in \order{\poly n}$ neighbors. Let $\operatorname{adj}: [2^n] \times [S] \to [2^n]$ be the mapping such that
\begin{equation}
    \operatorname{adj}(i, \ell) = \begin{cases}
        j, & 1 \le i \le N \text{ and } j \text{ is the } \ell\text{-th neighbor of } i \\
        0, & 1 \le i \le N \text{ and } \ell \ge \deg(i) \\
        \ell, & \text{otherwise} 
    \end{cases}.
\end{equation}
A sparse access for $G_s$ is given by the pair of oracles $O_\text{adj}, O_\text{sign}$,
\begin{align}
    O_\text{adj} \ket{i} \ket{\ell} & = \ket{\operatorname{adj}(i, \ell)} \ket{\ell}, \nonumber \\
    O_\text{sign} \ket{i} \ket{\ell} \ket{z} & = \ket{i} \ket{\ell} \ket{z \oplus s(\{i, \operatorname{adj}(i, \ell)\})}. \nonumber
\end{align}
\end{restatable}

The marked input model is nearly identical to the usual one, but it relaxes the requirement that the set of vertices must be \( \{ 1, \ldots, N\} \), allowing it instead to be any subset \( V \subseteq [2^n] \setminus \{0\} \). This difference is useful in the reduction from the \prbm{clique homology} task, as removing this requirement avoids the need to enumerate cliques. However, the same results can be proven in the traditional model with more effort.

\section{QMA1-hardness of \prbm{sparse balancedness}}\label{sec:balance}

This section is dedicated to proving Theorem~\ref{thm:thm_1}. In Section~\ref{sec:balance:def}, we explore key concepts in spectral \emph{signed} graph theory. Section~\ref{sec:balance:homology} introduces the \prbm{clique homology} problem. In Section~\ref{sec:balance:characterization}, we explain how the homological holes of a clique complex can be characterized in terms of the balancedness components of a signed graph. In Section~\ref{sec:balance:reduction}, we demonstrate that our characterization preserves efficient access to the corresponding Laplacian operators, thereby enabling a reduction from \prbm{clique homology} to the computational problem we define as \prbm{sparse balancedness}. 
In Section~\ref{sec:balance:model}, we show that our results extend to graphs expressed in the traditional sparse access.
Finally, in Section~\ref{sec:balance:containment} we prove that \prbm{sparse balancedness} under a suitable promise on the ground state energy of the signed Laplacian of the input graph is contained in \class{QMA}.

\subsection{Signed graphs and balancedness}\label{sec:balance:def}

We provide a brief introduction to signed graphs and their properties. An accessible treatment of this topic can be found in~\cite{dittrich2020signal}. Let $G_s = (V, E, s)$ be a signed graph. Without loss of generality, we assume the set of vertices is $V \subseteq [2^n] \setminus \{0\}, |V| = N$. We denote the real vector spaces of vertices and edges as $\calV^{G_s}$ and $\calE^{G_s}$, respectively:
\begin{align}
    \calV^{G_s} & = \operatorname{span}(V, \RR), \\
    \calE^{G_s} & = \operatorname{span}(E, \RR).
\end{align}

We omit the explicit reference to $G_s$ when it is clear from the context. The natural basis for $\calV$ is $\{\, \ket{i} \,\}_{i = 1}^N$, and for $\calE$ it is $\{ \,\ket{(i, j)} \,\}_{\{i, j\} \in E, i < j}$. The structure of the signed graph is encoded in the Laplacian operator. Analogous to the unsigned case, the spectral theory of the Laplacian of a signed graph provides deep insights into the network.

\restatabledefsignedlaplacian*

The Laplacian can also be equivalently defined in terms of the \emph{incidence operator} $\calN$.

\begin{proposition}
The operator $\calN: \calV \to \calE$ is defined as
\begin{equation}
    \mel{(i, j)}{\calN}{k} = \begin{cases}
        1, & i = j \\
        -\sng((i,j)), & i = k \\
        0, & \text{otherwise}
    \end{cases}
\end{equation}
and satisfies $\calN^\top \calN = \calL$. 
\end{proposition}
\begin{proof}
By direct calculation,
\begin{equation*}
    \mel{i}{\calN^\top \calN}{j}
    = \mel{i}{\calN^\top \mathds{1}_\calE \calN}{j}
    = \sum_{e \in E} \mel{i}{\calN^\top}{e} \!\! \mel{e}{\calN}{j}.
\end{equation*}
For $i = j$, we have
\begin{equation*}
    \sum_{e \in E} \mel{i}{\calN^\top}{e} \!\! \mel{e}{\calN}{i}
    = \sum_{i \sim k} \mel{i}{\calN^\top}{(i,k)} \!\! \mel{(i,k)}{\calN}{i}
    = \sum_{i \sim k} |\!\mel{i}{\calN^\top}{(i,k)}\!|^2
    = \sum_{i \sim k} 1 = \deg(i).
\end{equation*}
For $i \sim j, i < j$, there is a unique edge $\bar{e}$ connecting the vertices since the graph is simple. It follows 
\begin{equation*}
    \sum_{e \in E} \mel{i}{\calN^\top}{e} \!\! \mel{e}{\calN}{j}
    = \mel{i}{\calN^\top}{\bar e} \!\! \mel{\bar e}{\calN}{j} 
    = \mel{i}{\calN^\top}{(i,j)} \!\! \mel{(i,j)}{\calN}{j} 
    = 1 \cdot (-\sng((i,j))) = - \sng((i,j))
\end{equation*}
For $i \neq j$ and $i \not\sim j$, there are no edges connecting the vertices, so
\begin{equation*}
    \sum_{e \in E} \mel{i}{\calN^\top}{e} \!\! \mel{e}{\calN}{j}
    = 0. \qedhere
\end{equation*}
\end{proof}

A convenient aspect of using the Laplacian operator over other representations (e.g., adjacency matrix) is its positive semi-definiteness.

\begin{proposition}
$\calL$ is positive semi-definite.
\end{proposition}
\begin{proof}
For all $\ket{x} \in \calV$, it holds that
\begin{equation*}
    \mel{x}{\calL}{x} = \mel{x}{\calN^\top \calN}{x} = \norm{\calN \ket{x}}^2_2 \ge 0. \qedhere
\end{equation*}
\end{proof}

Note that this implies the (unsigned) graph Laplacian is positive semi-definite, too. One of the most important aspects of a signed network is its balancedness, which is defined as follows.

\begin{definition}\label{def:balanceness}
A signed graph is \emph{balanced} or \emph{structurally balanced} if there exists a partition of the vertices $V = V_1 \oplus V_2$ such that the edges within $V_1$ and within $V_2$ are all positive, while the edges between $V_1$ and $V_2$ are all negative.
\end{definition}

As previously mentioned, balancedness can be determined using linear algebraic methods based on the spectrum of $\calL$. Specifically, if the kernel of $\calL$ is non-trivial, its eigenvectors correspond to the indicator vectors of the partition. If the network is not balanced, the first non-zero eigenvalue is called the \emph{algebraic conflict} and serves as a measure of "how imbalanced" the network is~\cite{dittrich2020signal}. The role of the smallest non-zero eigenvalue is analogous to that of the Fiedler value in quantifying the connectedness in unsigned graphs.

\begin{proposition}\label{prop:balanced_kernel}
A \emph{connected} signed graph is balanced if and only if $\ker(\calL)$ is non-trivial.
\end{proposition}
\begin{proof}
See \textcite[Theorem 4.4]{kunegis2010spectral}.
\end{proof}

For signed graphs with possibly many connected components, $\calL$ admit a block diagonal form.

\begin{proposition}
Let $G_s$ be a signed graph with $k > 1$ connected components and $\calL^{G_s}$ its signed Laplacian. Then, $\calL^{G_s}$ admit a block diagonal form with $k$ blocks,
\begin{equation}
    \calL^{G_s'} = \mqty[ \calL^{G_s^{(1)}} & & \\ & \ddots & \\ & & \calL^{G_s^{(k)}} ].
\end{equation}
Here, $\mathcal{L}^{G_s^{(j)}}$ denotes the signed Laplacian of the subgraph induced by the $j$-th connected component of $G_s$.
\end{proposition}
\begin{proof}
See \textcite{kunegis2010spectral}.
\end{proof}

It follows that, for a signed graph $G_s$ with $k$ connected components whose kernel of its signed Laplacian can be stated as $\ker(\mathcal{L}^{G_s}) = \bigoplus_{j=1}^k \ker(\mathcal{L}^{G_s^{(j)}})$, we have $\ker(\mathcal{L}^{G_s}) \neq \varnothing$ if and only if at least one of its components is balanced.

\subsection{Clique homology}\label{sec:balance:homology}

In algebraic topology, homology is used to characterize the "holes" in a manifold, which are topological features that remain invariant under smooth deformations. Computationally, manifolds are approximated by structures called \emph{simplicial complexes} (which may be \emph{abstract} if geometrical information is disregarded), generalizing graphs. We briefly recall a few concepts from algebraic topology; for an introduction focused on the combinatorial aspect, refer to \textcite{lim2020hodge}. We recall the definition of ASC:

\restatableasc* 

In principle, the number of simplices can grow combinatorially with \( n \), making them difficult to represent explicitly. However, certain classes of ASCs allow for a more succinct representation. A relevant family in this context is that of clique complexes.

\begin{definition}
Let $G = (\{1, \ldots, n\}, E)$ be an undirected, unsigned graph. The \emph{clique complex} of $G$, denoted as $\clique(G)$, is an ASC with the set of vertices $V = \{1, \ldots, n\}$ and the set of simplices
\begin{equation}
\Sigma = \{ \varnothing \neq \sigma \subseteq V \mid \forall v, v' \in \sigma : \{v, v'\} \in E \}.
\end{equation}
\end{definition}

Let $\Gamma = (V, \Sigma)$ be an ASC. We can assign this object a linear algebraic structure called the linear chain space,
\begin{equation}
    \mathcal{C}^\Gamma = \operatorname{span}(\Sigma, \mathbb{R}).
\end{equation}
We will omit the explicit reference to $\Gamma$ when it is clear from the context. This space can be structured as a graded vector space,
\begin{equation}
    \mathcal{C}^\Gamma = \bigoplus_{p=0}^n \mathcal{C}_p^\Gamma,
\end{equation}
where $\mathcal{C}_p^\Gamma = \operatorname{span}(\Sigma_p, \mathbb{R})$ is the space of $p$-chains of $\Gamma$. A natural basis for $\mathcal{C}_p$ consists of the elementary $p$-chains, each corresponding to a single simplex. The interest in this algebraic representation lies in the operators that can be defined on it, particularly the boundary operator.

\begin{definition}
Let $\Gamma = (V, \Sigma)$ be an ASC and $\mathcal{C}^\Gamma = \bigoplus_{p=0}^n \mathcal{C}_p^\Gamma$ its chain space. The $p$-boundary operator $\partial_p^\Gamma: \mathcal{C}_p^\Gamma \to \mathcal{C}_{p-1}^\Gamma$, for $p = 0, \ldots, n$, is defined on the natural basis of $\mathcal{C}_p^\Gamma$ as
\begin{equation}
\partial_p^\Gamma \sigma = \sum_{j=0}^p (-1)^j (\sigma \setminus \{ \sigma_j \}).
\end{equation}
Here, $\sigma_j$ denotes the $j$-th element of $\sigma$. To refer to the $j$-th element of a set, we consider the simplices ordered with respect to the natural ordering of their vertices, $\sigma = (v_0, \ldots, v_p)$ with $v_0 < \ldots < v_p$.
\end{definition}

A key property of the boundary operator is that the boundary of a boundary is empty.

\begin{proposition}
$\partial_{p+1} \circ \partial_p = 0$, i.e., $\operatorname{im}(\partial_{p+1}) \subseteq \ker(\partial_p)$.
\end{proposition}
\begin{proof}
See \textcite[Theorem B.4]{lim2020hodge}.
\end{proof}

The $p$-cycles are elements of $\ker(\partial_p)$, i.e., $p$-chains with zero boundary. Similarly, a $p$-boundary is a $p$-cycle that belongs to the subset $\operatorname{im}(\partial_{p+1}) \subseteq \ker(\partial_p)$, i.e., a $p$-cycle that is also the boundary of a $(p+1)$-chain. A $p$-hole is an element of $\ker(\partial_p) / \operatorname{im}(\partial_{p+1})$, i.e., an equivalence class of $p$-cycles that are not $p$-boundaries. The quotient space $\ker(\partial_p) / \operatorname{im}(\partial_{p+1})$ is called the $p$-homology space. Now, we are ready to define the computational problem at the core of our reduction.

\begin{problem} \prbm{clique homology}

\noindent{}Input: graph $G = (\{1, \ldots, n\}, E)$, $p \in \{0, \ldots, n\}$.

\noindent{}Output: \yes{} if $\ker(\partial_p^{\clique(G)}) / \operatorname{im}(\partial_{p+1}^{\clique(G)})$ is non-trivial, \no{} otherwise.
\end{problem}

Recent work has focused on characterizing the hardness of this computational problem. The result we will use is stated in the following proposition:

\begin{proposition}\label{thm:clique_homology_qmahard}
\prbm{clique homology} is \qmaone{}-hard.
\end{proposition}
\begin{proof}
See \textcite{crichigno2022clique}.
\end{proof}

The work of \textcite{friedman1996computing} established a deep connection between algebraic topology and Hodge theory, showing that a linear operator known as the Combinatorial Hodge Laplacian (defined over an ASC) encodes the $p$-holes in its kernel. This has paved the way for the application of linear algebraic tools in algebraic topology.

\begin{definition}
Let $\Gamma = (V, \Sigma)$ be an ASC with $\deg \Gamma = \max_{\sigma \in \Sigma} \abs{\sigma} - 1$ being the maximum degree of the simplices in $\Gamma$. For $p = 0, \ldots, \deg \Gamma$, the $p$-\emph{Combinatorial Hodge Laplacian} is the linear operator $\calL_p^{\Gamma}: \mathcal{C}_p^\Gamma \to \mathcal{C}_p^\Gamma$ defined as
\begin{equation*}
    \calL_p^{\Gamma} = \begin{cases}
    \partial_{1} \partial_1^\dagger & p = 0 \\
    \partial_{p}^\dagger \partial_p + \partial_{p+1} \partial_{p+1}^\dagger & p = 1, \ldots, \deg \Gamma - 1 \\
    \partial_{p}^\dagger \partial_p & p = \deg \Gamma.
    \end{cases}
\end{equation*}
\end{definition}

The following proposition states that the $p$-holes of the abstract simplicial complex are encoded in the kernel of its Combinatorial Hodge Laplacian.

\begin{proposition}{\emph{(\cite{friedman1996computing})}}\label{thm:friedman}
The kernel of $\mathcal{L}_p^{\Gamma}$ is isomorphic to $\ker{\partial_p} / \operatorname{im}(\partial_{p+1})$.
\end{proposition}

\subsection{Characterization of homology via balancedness}\label{sec:balance:characterization}

Abstract simplicial complexes share some similarities with signed graphs. These similarities have been explored in the literature, notably in the work of \textcite{she2019algebraic}, which uses algebraic topology related to simplicial homologies to study the balancedness of signed graphs. In this work, we take the opposite approach by characterizing simplicial homologies as the balanced components of a signed graph.

We begin by recalling the definition of the Combinatorial Hodge Laplacian and highlighting its decomposition into two distinct components:
\begin{equation}\label{eq:decomposition} 
    \calL_p = \begin{cases}
        \partial_{1} \partial_1^\dagger, & p = 0 \\
        \partial_{p}^\dagger \partial_p + \partial_{p+1} \partial_{p+1}^\dagger, & p = 1, \ldots, \deg \Gamma - 1 \\
        \partial_{p}^\dagger \partial_p, & p = \deg \Gamma
    \end{cases}
    = \begin{cases}
        \phantom{\calL_p^{\downarrow} +\;\;} \calL_p^{\uparrow}, & p = 0 \\
        \calL_p^{\downarrow} + \calL_p^{\uparrow}, & p = 1, \ldots, \deg \Gamma - 1 \\
        \calL_p^{\downarrow} \phantom{+ \calL_p^{\uparrow}\;\,}, & p = \deg \Gamma
    \end{cases}.
\end{equation} 

The two components, $\calL_p^{\downarrow}$ and $\calL_p^{\uparrow}$, encode different types of information about the $p$-simplices. Specifically, $\calL_p^{\downarrow}$, known as the \emph{lower Laplacian}, encodes the lower-adjacency relation, where $\sigma \sim_l \tau$ if and only if they share a face. Similarly, $\calL_p^{\uparrow}$, called the \emph{upper Laplacian}, encodes the upper-adjacency relation, where $\sigma \sim_u \tau$ if and only if they are faces of a common $(p+1)$-simplex $\rho$. The number of upper-adjacent simplices to a given $p$-simplex $\sigma$ is denoted by $\deg^\uparrow(\sigma)$.

The work of \textcite{jost2023cheeger} has characterized a variant of the upper component of the upper $p$-Combinatorial Hodge Laplacian over some ASC $\Gamma = (V, \Sigma)$ in terms of a variant of the signed Laplacian over a signed graph whose vertices correspond to the $p$-simplices $\Sigma_p$. We extend this construction by removing the normalization factor and modifying their approach so that an identical result holds for the lower-Laplacian component. Finally, we show that adding the matrix representations of the signed Laplacians of the two components results in a new matrix, which is again the signed Laplacian of yet another signed graph.

To build this construction, we rely on the following definition:

\begin{definition}\label{def:sng}
Let $(V, \Sigma)$ be an ASC, $\sigma, \rho \in \Sigma$, and $\sigma = \{v_0, \ldots, v_{j-1}, v_{j+1}, \ldots, v_p\}$ be a face of $\rho = \{v_0, \ldots, v_p\}$. The \emph{sign} of $\sigma$ with respect to $\rho$ is defined as the mapping
\begin{align} 
\sng & : \Sigma_{p-1} \times \Sigma_{p} \to \{ 0, \pm 1 \}, \\
\sng & (\sigma, \rho) = \begin{cases}
    (-1)^j, & \sigma \text{ is the } j\text{-th face of } \rho \\
    0, & \text{otherwise}.
\end{cases} 
\end{align} 
The boundary operator $\partial_{p}$ can then be written as
\begin{equation}
    \mel{\sigma}{\partial_{p}}{\rho} = \sng(\sigma, \rho).
\end{equation}
\end{definition} 

Then, we prove the following proposition:

\restatablepropascsigned*
\begin{proof}
Let $\Gamma = (V, \Sigma)$ be an ASC and $p \in \NN$. We can deal with the upper- and lower- portion of the Combinatorial Hodge Laplacian separately. It follows from a direct calculation that  
\begin{align} 
\mel{\tau}{\calL_p^{\uparrow, \Gamma}}{\sigma} & = \mel{\tau}{\partial_{p+1} \partial_{p+1}^\top}{\sigma} \nonumber \\
& = \sum_{\rho \in \Sigma_{p+1}} \mel{\tau}{\partial_{p+1}}{\rho} \!\! \mel{\rho}{\partial_{p+1}^\top}{\sigma} \\
& = \sum_{\rho \in \Sigma_{p+1}}  \sng(\sigma, \rho) \sng(\tau, \rho) \nonumber \\ 
& = \begin{cases} 
    \sum_{\rho \in \Sigma_{p+1}} (\sng(\sigma, \rho))^2 = \deg^{\uparrow}(\sigma), & \sigma = \tau \\
    \sum_{\rho \in \Sigma_{p+1}} \sng(\sigma, \rho) \sng(\tau, \rho) = \sng(\sigma, \bar\rho) \sng(\tau, \bar\rho), & \sigma \sim_u \tau \\
    0, & \text{otherwise} 
\end{cases}. 
\end{align} 
Here, $\bar\rho$ is the unique $p+1$ simplex that the two faces might share. Let $G_s^\uparrow = (V_s^\uparrow, E_s^\uparrow, s^\uparrow)$ with
\begin{align} 
V_s^\uparrow & = \Sigma_p, \nonumber \\ E_s^\uparrow & = \{ \{ \sigma, \tau \} \mid \sigma, \tau \in \Sigma_p, \sigma \sim_u \tau \}, \nonumber \\ s^\uparrow(\sigma, \tau) & = - \sng(\sigma, \bar\rho) \sng(\tau, \bar\rho).
\end{align}
Note that $\calC_p^\Gamma = \calV_s^\uparrow = \operatorname{span}(\Sigma_p^\Gamma, \RR)$. Consider the unitary bijection $\iota: \calC_p^\Gamma \to \calV_s^\uparrow, \iota(\ket{\sigma}) = \ket{v_\sigma}$, which maps $p$-chains of the ASC to vertices of our construction. From Definition~\ref{def:signed_laplacian} and direct calculation, it follows that:
\begin{align}
    & \bra{\tau} \calL_p^{\uparrow, \Gamma} \ket{\sigma} \nonumber \\
    & = \iota(\bra{\tau}) \; \iota(\calL_p^{\uparrow, \Gamma} \ket{\sigma}) \nonumber \\
    & = \iota(\bra{\tau}) \; \iota \left(
        \sum_{\alpha, \beta \in \Sigma_p} (\deg^\uparrow(\alpha) \delta_{\alpha\beta} + \sng(\alpha,\bar\rho)\sng(\beta,\bar\rho)\delta_{\alpha\sim\beta}) \ketbra{\alpha}{\beta} \cdot \ket{\sigma}
    \right) \nonumber \\
    & = \bra{v_\tau} \; \sum_{v_\alpha, v_\beta \in V_s^\uparrow} (\deg(v_\alpha) \delta_{v_\alpha v_\beta} -s^\uparrow(v_\alpha, v_\beta) \delta_{v_\alpha\sim v_\beta}) \ketbra{v_\alpha}{v_\beta} \cdot \ket{v_\sigma}
     \nonumber \\
    & = \bra{v_\tau} \calL^{G_s^\uparrow} \ket{v_\sigma}.
\end{align}
Similarly, 
\begin{align} 
\mel{\tau}{\calL_p^{\downarrow, \Gamma}}{\sigma}
& = \sum_{\nu \in \Sigma_{p-1}} \mel{\tau}{\partial_{p}^\top}{\nu} \mel{\nu}{\partial_{p}}{\sigma} \nonumber \\ 
& = \begin{cases} 
\sum_{\nu \in \Sigma_{p-1}} (\sng(\nu, \sigma))^2, & \sigma = \tau \\ \sum_{\nu \in \Sigma_{p-1}} \sng(\nu, \sigma) \sng(\nu, \tau), & \sigma \neq \tau \end{cases} \nonumber \\ 
& = \begin{cases} 
p + 1, & \sigma = \tau \\ 
\sng(\sigma \cap \tau, \sigma) \sng(\sigma \cap \tau, \tau), & \sigma \sim_l \tau \\ 
0, & \text{otherwise} 
\end{cases}. \nonumber 
\end{align} 
Here, $\sigma \cap \tau$ is the only possible $(p-1)$-simplex whose product of $\sng$ is non-zero. Furthermore, the number of $(p-1)$-simplices lower-adjacent is always $p+1$. This is obtained by considering the faces of the $p$-simplices that, by the closure by inclusion of the set of simplices in the ASC, are always present. Let $G_s^\downarrow = (V_s^\downarrow, E_s^\downarrow, s^\downarrow)$ with: 
\begin{align} 
V_s^\downarrow & = \Sigma_p, \nonumber \\ 
E_s^\downarrow & = \{ \{ \sigma, \tau \} \mid \sigma, \tau \in \Sigma_p, \sigma \sim_l \tau \}, \nonumber \\ 
s^\downarrow(\sigma, \tau) & = - \sng(\sigma \cap \tau, \partial \sigma) \sng(\sigma \cap \tau, \partial \tau). 
\end{align} 
From Definition~\ref{def:signed_laplacian} and a direct calculation similar to that in the previous proposition, we have that $\calL_p^{\downarrow, \Gamma}$ and $\calL^{G_s^\downarrow}$ are unitarily equivalent. Finally, we can prove that the sum of these components results in the signed Laplacian of yet another signed graph, capturing a different relationship: \emph{either $\sigma \sim_u \tau$ or $\sigma \sim_l \tau$}. For Equation~\ref{eq:decomposition}, it holds that $\calL_p^\Gamma = \calL_p^{\downarrow, \Gamma} + \calL_p^{\uparrow, \Gamma} = \calL^{G_s^\uparrow} + \calL^{G_s^\downarrow}$. However, it might not be the case that $\calL^{G_s^\uparrow} + \calL^{G_s^\downarrow}$ correctly describes the Laplacian of yet another signed graph, as the off-diagonal entries are in the range $\{0, \pm 1, \pm 2\}$. The explicit form of this sum of operators results in 
\begin{equation*} 
\mel{\tau}{\calL^{G_s^\uparrow} + \calL^{G_s^\downarrow}}{\sigma}
= \begin{cases} 
    p + 1 + \deg^\uparrow(\sigma), & \sigma = \tau \\ 
    - \mathbf{s}(\sigma, \tau) & \sigma \sim_u \tau \text{ or } \sigma \sim_l \tau \\ 
    0, & \text{otherwise} 
\end{cases}
\end{equation*} 
where 
\begin{align*} 
\mathbf{s}(\sigma, \tau) = - & \Big( \sng(\sigma \cap \tau, \partial \sigma) \times \sng(\sigma \cap \tau, \partial \tau) + \sng(\sigma, \partial \rho) \times \sng(\tau, \partial \rho) \Big). 
\end{align*} 
We can show that $\mathbf{s}$ is a proper sign function. For $\sigma \not\sim_l \tau$ and $\sigma \not\sim_u \tau$, both terms in $\mathbf{s}$ are zero by direct calculation. For either $\sigma \sim_l \tau$ or $\sigma \sim_u \tau$, exactly one term is non-zero and has a value in $\{ \pm 1\}$. If both $\sigma \sim_l \tau$ and $\sigma \sim_u \tau$ hold simultaneously, we can prove the term is zero by direct calculation. Without loss of generality, let $0 \le i < j \le p+1$ and 
\begin{align*} 
\rho & = \{ v_0, \ldots, v_{p+1} \}, \\
\sigma & = \{ v_0, \ldots, v_{i-1}, v_{i+1}, \ldots, v_{p+1} \}, \\ 
\tau & = \{ v_0, \ldots, v_{j-1}, v_{j+1}, \ldots, v_{p+1} \}, \\ 
\sigma \cap \tau & = \{ v_0, \ldots, v_{i-1}, v_{i+1}, \ldots, v_{j-1}, v_{j+1}, \ldots, v_{p+1} \}.
\end{align*} 
Then, 
\begin{align*} 
\sng(\sigma \cap \tau, \partial \sigma) \times \sng(\sigma \cap \tau, \partial \sigma) & = (-1)^i \times (-1)^{j-1}, \\ 
\sng(\sigma, \partial \rho) \times \sng(\tau, \partial \rho) & = (-1)^i \times (-1)^{j}. 
\end{align*} 
Thus, the sum of the two terms is zero. Once we have proven that $\mathbf{s}$ is a proper sign function, the two operators are unitarily equivalent as a sum of unitarily equivalent pairs of operators.
\end{proof}

\subsection{Reduction}\label{sec:balance:reduction}

We recall the definition of the balancedness task for graphs admitting an efficient implementation of its sparse access.

\restatablesparsebalancedness*

\restatabletheoremone*

\begin{proof} 
We will demonstrate hardness via a reduction from \prbm{clique homology}.

\medskip\noindent\textbf{Reduction}:

\begin{enumerate}
    \item The input to \prbm{clique homology} is a graph $G = (\{1, \ldots, n\}, E)$ and an integer $p \in \NN$.
    \item We define the classical circuit $O_\text{adj}$ implementing the following algorithm: \\
    
    \begin{algorithmic}[1]
\Statex \textbf{Function} $O_\text{adj}(i, \ell)$
\Require $i \in 0, \ldots, 2^n-1$
\Require $\ell \in 0, \ldots, 2^m-1, m = \lceil \log_2 (n-p-1)(p+1) \rceil$
\State $\sigma \gets \{ k \mid \operatorname{bitstring}(i)[k] = 1 \}$
\If{($|\sigma| \neq p+1$) or ($\sigma$ is not a clique in $G$)} \Comment{$i$ is not a vertex of $G_s$, $\order{n^2}$}
    \State \textbf{return} $\ell$
\Else \Comment{A vertex of $G_s$, corresponding also to a clique of size $p+1$ in $G$}
    \State $V' = V \setminus \sigma$ \Comment{Set of vertices not belonging to the current clique, $|V'| = n - p - 1$}
    \State $a = \ell \text{ rem } (p + 1)$ \Comment{$a \in \{0, ..., p\}$}
    \State $b = \lfloor \ell / (p + 1) \rfloor$ \Comment{$b \in \{0, ..., n-p-1\}$}
    \State $\sigma' \gets \sigma $ minus its $a$-th vertex \Comment{$\order{n}$}
    \State $\sigma'' \gets \sigma' $ plus the $b$-th vertex in $V'$ \Comment{$\order{n}$}
    \State \textbf{return} bitstring representation of $\sigma''$ \Comment{Possibly a clique in $G$, will be verified afterwards}
\EndIf
\end{algorithmic}

    \item We define the classical circuit $O_\text{sign}$ implementing the following algorithm: \\

    \begin{algorithmic}[1]
\Statex \textbf{Function} $O_\text{sign}(i, \ell)$
\Require $i \in 0, \ldots, 2^n-1$
\Require $\ell \in 0, \ldots, 2^m-1, m = \lceil \log_2 n^2 \rceil$
\State $j \gets O_\text{adj}(i, \ell)$
\State $\sigma \gets \{ k \mid \operatorname{bitstring}(i)[k] = 1 \}$
\State $\tau \gets \{ k \mid \operatorname{bitstring}(j)[k] = 1 \}$
\If{$j = 0$ or ($|\tau| \neq p+1$) or ($\tau$ is not a clique in $G$)}
    \State \textbf{return} 0
\Else
    \State \textbf{return} $-\sng(\sigma\cap\tau,\sigma)\sng(\sigma\cap\tau,\tau) -\sng(\sigma,\sigma\cup\tau)\sng(\tau,\sigma\cup\tau)$ \Comment{$\sng$ costs $\order{n}$}
\EndIf
\end{algorithmic}

    \item The circuits $(O_\text{adj}, O_\text{sign})$ implement sparse access as per Definition \ref{def:bounded_degree_signed}. The circuit for $O_\text{adj}$ meets the three conditions in the definition of sparse access: the placeholder index $0$ has the adjacency list $[0, \dots, 0]$, while indices not associated with vertices in $G_s$ (either bitstrings with Hamming weight different from $p+1$, corresponding in our encoding to simplices of dimension different from $p$, or bitstrings correctly addressing a subset of $p+1$ vertices that do not form a clique) are assigned the placeholder adjacency list $[0, 1, \dots, S-1]$. The circuit for $O_\text{sign}$ assigns a non-zero value only to pairs of indices that actually correspond to vertices of the graph.

    \item The circuits $(O_\text{adj}, O_\text{sign})$ implement sparse access for the graph $G_s$ obtained via Proposition \ref{prop:full_laplacian}. The circuit for $O_\text{sign}$ assigns a non-zero value only to pairs of indices that correspond to vertices of the graph, and the assigned value matches the sign function of the graph as described in Proposition \ref{prop:full_laplacian}.

    \item We solve \prbm{sparse balancedness} with $O_\text{adj}$ and $O_\text{sign}$ as input.
\end{enumerate}

\medskip\noindent\textbf{Completeness}: 

\begin{enumerate}
    \item Let $(G, p)$ be a \yes{} instance of \prbm{clique homology}. \item For Theorem~\ref{thm:friedman}, $\ker\calL_p^{\clique(G)} \neq \varnothing$.
    \item For Proposition~\ref{prop:full_laplacian}, there exists a signed graph $G_s$ such that $\calL^{G_s}$ and $\calL_p^{\clique(G)}$ are unitarily equivalent, and the classical circuits $(O_\text{adj}, O_\text{sign})$ defined in the reduction above implement its sparse access. 
    \item $\ker\calL^{G_s} \neq \varnothing$. For Theorem \ref{prop:balanced_kernel}, $(O_\text{adj}, O_\text{sign})$ is a \yes{} instance of \prbm{sparse balancedness}.
\end{enumerate}

\medskip\noindent\textbf{Soundness}: 

\begin{enumerate}
    \item Let $(G, p)$ be a \no{} instance of \prbm{clique homology}. \item For Theorem~\ref{thm:friedman}, $\ker\calL_p^{\clique(G)} = \varnothing$.
    \item For Proposition~\ref{prop:full_laplacian}, there exists a signed graph $G_s$ such that $\calL^{G_s}$ and $\calL_p^{\clique(G)}$ are unitarily equivalent, and the classical circuits $(O_\text{adj}, O_\text{sign})$ defined in the reduction above implement its sparse access.
    \item $\ker\calL^{G_s} = 0$. For Theorem \ref{prop:balanced_kernel}, $(O_\text{adj}, O_\text{sign})$ is a \no{} instance of \prbm{sparse balancedness}.
\end{enumerate}
\end{proof}

Notably, we could have proven Theorem \ref{thm:thm_1} could have been proven even with a slightly simpler statement than the one in Proposition~\ref{prop:full_laplacian}: for each symmetric matrix with off-diagonal entries in $\{-1, 0, 1\}$ and diagonal counting the non-zero entries per row there is a signed graph whose signed Laplacian matches the given matrix. The Combinatorial Hodge Laplacian is an instance of such symmetric matrices.

\subsection{Role of the marked sparse access}\label{sec:balance:model}

Our results for the marked sparse access model can be translated back to the traditional sparse access model. Specifically, for any signed graph \( G_s \), we can construct a signed graph \( G_s' \) with a comparable number of vertices and an upper bound on their degree such that \( G_s \) has a balanced component if and only if \( G_s' \) has a balanced component.

This result is achieved informally by augmenting \( G_s = (V, E, s) \) with vertices associated with bitstrings \( i \not\in V \), along with \( \lceil 2^n / S \rceil + 3 \) auxiliary vertices. These auxiliary vertices form a new connected component, disconnected from the existing components in \( G_s \), and we ensure that this new component is unbalanced so that it does not contribute additional elements to the kernel of \( \mathcal{L}^{G_s'} \). We confirm that this construction is efficient.

This result emphasizes that the marked sparse access model is, for this task, merely a convenient way to represent our graph without adding any new capability beyond what is available in the traditional sparse access model. This statement is not guaranteed to hold for other tasks, such as \prbm{sparse balancedness} in case the input is restricted to connected graphs or \prbm{guided local hamiltonian} \cite{cade_et_al} (in case the Hamiltonian corresponds to a Combinatorial Laplacian).

\begin{proposition}\label{prop:equiv_marked_access}
Let $G_s = (V, E, s)$ be a signed graph with $V \subseteq [2^n] \setminus \{0\}, |V| = N$, such that
\begin{itemize}
    \setlength\itemsep{0em}
    \item the graph is sparse, i.e., there is an upper bound $2 \le S \in \order{\poly n}$ on the degree of the vertices;
    \item the graph is given as a classical circuit of size $\order{\poly n}$ implementing the marked sparse access $O_\text{adj}, O_\text{sign}$.
\end{itemize} 
Then, there exists a signed graph $G_s' = (V', E', s')$ vertices, $V' = \{1, ..., N'\}$, such that
\begin{itemize}
    \setlength\itemsep{0em}
    \item $N' \le 2^{n+1}$;
    \item the upper bound on the degree of the vertices is $S + 2\in \order{\poly n}$;
    \item $G_s$ has a balanced component iff $G_s'$ has a balanced component. 
    \item there exists a pair of circuits $O_\text{adj}', O_\text{sign}'$ implementing the traditional sparse access and uses $\order{1}$ calls to $O_\text{adj}, O_\text{sign}$;
\end{itemize} 
\end{proposition}
\begin{proof}

Let $A = \lceil 2^n/S \rceil$. Define $G_s' = (V', E', s')$ such that
\begin{align}
    V' & = \{1, ..., N'\} \text{ with } N' = 2^n + A + 3 \\
    E' & = E & (\text{all the edges in } G_s) \\
    & \phantom{=} \cup \{ \{i, i+1\} \mid i = 2^n, \ldots, 2^n+A \} & (\text{auxiliary vertices form a line}) \\
    & \phantom{=} \cup \{ \{ i, 2^n + \lceil i / S \rceil \} \mid i \in [2^n], i \not\in V \} & (\text{bitstring } i \not\in V \text{ connected to their nearest auxiliary vertex}) \\
    & \phantom{=} \cup \{ \{2^n+A+1, 2^n+A+2\}, \\
    & \quad \phantom{=} \{2^n+A+2, 2^n+A+3\}, \\
    & \quad \phantom{=} \{2^n+A+1, 2^n+A+3\} \} & (\text{triangle})
\end{align}
and
\begin{equation}
    s'(\{i, j\}) = \begin{cases}
        s(\{i, j\}) & i, j < 2^n \\
        -1 & \{i, j\} = \{ 2^n+A+2, 2^n+A+3\} \\
        +1, & \text{otherwise}
    \end{cases}
\end{equation}

By construction, \( G_s' \) has fewer than \( 2^{n+1} \) vertices for \( S \ge 2 \), and each auxiliary vertex has no more than \( S + 2 \) neighbors: two vertices preceding and succeeding it in the line, and \( S \) additional vertices corresponding to bitstrings not associated with vertices of \( G_s \). 

Denoting \( G_s^\text{new} \) as the subgraph composed solely of the augmented vertices (auxiliary ones and the triangle), we note that this subgraph is connected by construction, without connections to any \( i \in V \). Consequently, the Laplacian of \( G_s' \) can be written as
\begin{equation}
    \calL^{G_s'} = \mqty[ \calL^{G_s} & \mathbf{0} \\ \mathbf{0} & \calL^{G_s^\text{new}} ].
\end{equation}
The subgraph \( G_s^\text{new} \) is unbalanced because it contains a path with an odd number of negative edges (\(\{2^n+A+1, 2^n+A+2\}\), \(\{2^n+A+2, 2^n+A+3\}\), \(\{2^n+A+3, 2^n+A+1\}\)), and thus this component is never balanced, yielding \( \ker \calL^{G_s^\text{new}} = \varnothing \). This implies that \( \ker \calL^{G_s'} = \ker \calL^{G_s} \oplus \ker \calL^{G_s^\text{new}} = \ker \calL^{G_s} \).

Finally, we show that the traditional sparse access \( O_\text{adj}', O_\text{sign}' \) can be efficiently implemented from the marked sparse access model using the following algorithms.

\begin{algorithmic}[1]
\Statex \textbf{Function} $O_\text{adj'}(i, \ell)$
\Require $i \in 0, \ldots, 2^{n+1}-1$
\Require $\ell \in 0, \ldots, S +2 - 1$
\If{$i < 2^n$} \Comment{Vertices from the original graph $G_s$}
    \If{$i \not\in V$} \Comment{Not a vertex of $G_s$}
        \If{$\ell = 0$} \Comment{The vertex $i$ will have as its only neighbor the vertex $\lceil i/S \rceil$ in the second half of $V'$}
            \State \textbf{return} $2^n + \lceil i/S \rceil$ 
        \Else
            \State \textbf{return} $0$
        \EndIf
    \ElsIf{$i \in V$} \Comment{A vertex of $G_s$}
        \State \textbf{return} $O_\text{adj}(i, \ell)$
    \EndIf
\ElsIf{$2^n \le i \le 2^{n} + \lceil 2^n/S \rceil$} \Comment{Auxiliary vertices}
    \If{$\ell = 0$} \Comment{The first neighbor is the next in the chain of new vertices}
        \State \textbf{return} $i+1$
    \ElsIf{$i \neq 2^n$ and $\ell = 1$} \Comment{The second neighbor is the next in the chain of new vertices (except $2^n$)}
        \State \textbf{return} $i-1$
    \Else \Comment{The other are the non-vertices in the original graph} 
        \State $i' \gets i - 2^n$ \Comment{First of the bitstring in $[2^n]$ that might be associated with this vertex}
        \State $\ell' \gets 2 - \delta_{i, 2^n}$ \Comment{counter of neighbors, $\ell = 2$ for anyone except $\ell = 1$ for $i=2^n$}
        \For{$j = 0, \ldots, S-1$}
            \If{$(i' + j) \not\in V$} \Comment{$j$-th potential bistring associated with this vertex}
                \State $\ell' \gets \ell + 1$
            \EndIf
            \If{$\ell' = \ell$} \Comment{return the $\ell+1$ bitstring $\not\in V$}
                \State \textbf{return} $i' + j$
            \EndIf
        \EndFor
        \State \textbf{return} $0$
    \EndIf
\ElsIf{$2^{n} + \lceil 2^n/S \rceil + 1 \le i \le 2^{n} + \lceil 2^n/S \rceil + 3$} \Comment{Triangle}
    \If{$i = 2^{n} + \lceil 2^n/S \rceil + 1$}
        \State \textbf{return} $[i-1, i+1, i+2, 0, \ldots, 0]_\ell$
    \ElsIf{$i = 2^{n} + \lceil 2^n/S \rceil + 2$}
        \State \textbf{return} $[i-1, i+1, 0, \ldots, 0]_\ell$
    \ElsIf{$i = 2^{n} + \lceil 2^n/S \rceil + 3$}
        \State \textbf{return} $[i-1, i-2, 0, \ldots, 0]_\ell$
    \EndIf
\Else
    \State \textbf{return} $\ell$
\EndIf
\end{algorithmic}

\bigskip
\begin{algorithmic}[1]
\Statex \textbf{Function} $O_\text{sign}'(i, \ell)$
\Require $i \in 0, \ldots, 2^{n+1}-1$
\Require $\ell \in 0, \ldots, S +2 - 1$
\State $j \gets O_\text{adj}'(i, \ell)$
\If{$j = 0$}
    \State \textbf{return} $0$
\ElsIf{$i \in V \text{ and } j \in V$}
    \State \textbf{return} $O_\text{sign}(i, \ell)$
\ElsIf{$\{i, j\} = \{ 2^{n} + \lceil 2^n/S \rceil + 2, 2^{n} + \lceil 2^n/S \rceil + 3\}$}
    \State \textbf{return} $-1$
\Else
    \State \textbf{return} $+1$
\EndIf
\end{algorithmic}

\end{proof}

\subsection{Containment in \class{QMA}}\label{sec:balance:containment}

We start by recalling the definition of the block-encoding of a linear operator.

\begin{definition}{(Block encoding~\cite{camps2024explicit})}
Let $A \in \RR^{N \times N}$ for $N = 2^n$, and let $\alpha, \epsilon \in \RR_{\ge 0}$. A unitary $U_A$ over $m + n$ qubits is a $(\alpha, m, \epsilon)$-block encoding of $A$ if
\begin{equation*}
    \norm{A - \alpha (\bra{0^m} \otimes \mathbb{I}_2^{\otimes n}) U_A (\ket{0^m} \otimes \mathbb{I}_2^{\otimes n})}_2 \le \epsilon.
\end{equation*}
\end{definition}

A block encoding for sparse matrices, such as those for the Laplacian of sparse graphs and the Combinatorial Hodge Laplacian, can be obtained with the following approach:

\begin{proposition}{(Block encoding for sparse matrices~\cite{camps2024explicit})}\label{def:sparse_access_be}
Let $A \in \RR^{N \times N}, \norm{A}_2 \le 1$ be an $S$-sparse symmetric matrix, $N = 2^n, S = 2^m$. If there exist a unitary $O_\text{row}$ over $n + s$ qubits and a unitary $O_\text{entry}$ over $1 + n + s$ qubits implementing
\begin{align}
    O_\text{row} \ket{i} \ket{\ell} & = \ket{c(i, \ell)} \ket{\ell}, \nonumber \\
    O_\text{entry} \ket{0} \ket{i} \ket{\ell} & = \left(A_{i,c(i,\ell)} \ket{0} + \sqrt{1 - |A_{i,c(i,\ell)}|^2}\right) \ket{i} \ket{\ell}, \nonumber
\end{align}
then there exists a $(1, m, 0)$-block encoding of $A$ that performs $\order{1}$ calls to $O_\text{row}$ and $O_\text{entry}$.
\end{proposition}

The matrix $A$ must be scaled to ensure the existence of the block encoding. In fact, the singular values of any submatrix block of a unitary matrix have to be $\leq 1$ \cite{camps2024explicit}. We recall the variant of \prbm{sparse balancedness} under a promise on the smallest eigenvalue of the Laplacian.

\restatablepromisesparsebalancedness*

We can prove that such a problem is contained in \class{QMA}. 

The quantum algorithm solving the task presents some subtleties regarding how we block encode the signed Laplacian operator. 
Specifically, we are \emph{embedding} the Laplacian $\calL^{G_s}$ as an operator acting on the Hilbert space of an $n$-qubit system. Notably, $\dim(\calL^{G_s}) = N < 2^n$, while the Hilbert space of the quantum system has dimension $2^n$. This may create some challenges: the kernel of the embedding operator $H$, which lives in the Hilbert space of dimension $2^n$, is not guaranteed to contain the same elements of the kernel of the operator $\calL^{G_s}$. To tackle this issue, we rely on the block-encoding of the following operator. 

\begin{proposition}\label{prop:ham_signed}
Let $G_s = (V, E, s)$ be a signed graph with $N \leq 2^n - 1$ vertices. The Hamiltonian $H = \calL^{G_s} + \sum_{i \not\in V} \ketbra{i}$ has dimension $\dim H = 2^n$ and kernel $\ker(H) = \ker(\calL^{G_s})$.
\end{proposition}

\begin{proof}
We can express the Hilbert space $\calH$ over $n$ qubits as the direct sum of the subspace spanned by $V$, the vertices of $G_s$, and the subspace spanned by $[2^n] \setminus V$. That is,
\begin{equation}
    \calH = \calH_V \oplus \calH_{\lnot V}.
\end{equation}

Define the Hamiltonian $H = \calL^{G_s} \oplus H'$, where $\calL^{G_s}$ acts non-trivially only on $\calH_V$ and $H'$ acts non-trivially only on $\calH_{\lnot V}$. Thus, we have $\ker(H) = \ker(\calL^{G_s}) \oplus \ker(H')$. Define $H' = \mathbb{I}_{\lnot V} = \sum_{i \not\in V} \ketbra{i}$. Since $\ker(\mathbb{I}_{\lnot V}) = \varnothing$, it follows that $\ker(H) = \ker(\calL^{G_s})$.

Furthermore, $\dim(H) = \dim(\calL^{G_s}) + \dim(\mathbb{I}_{\lnot V}) = N + 2^n - N = 2^n$. 
\end{proof}

Notably, the circuit implementing the marked sparse access allows us to obtain a block-encoding of the Hamiltonian in Proposition \ref{prop:ham_signed}.

\begin{proposition}\label{prop:efficient_be_laplacian}
Let $G_s$ be a signed graph with $N \le 2^n - 1$ vertices, $V \subseteq [2^n] \setminus \{0\}$, such that
\begin{itemize}
    \setlength\itemsep{0em}
    \item the graph is sparse, i.e., there is an upper bound $S \in \order{\poly n}$ on the degree of the vertices, $S \le 2^m$;
    \item the graph is given as a classical circuit of size $\order{\poly n}$ implementing the marked sparse access $(O_\text{adj}, O_\text{sign})$.
\end{itemize} 
Then, there exists a quantum circuit realizing a $(2S, m, 0)$-block encoding of $\calL^{G_s} + \sum_{i \not\in V} \ketbra{i}$ in time $\order{\poly n}$.
\end{proposition}
\begin{proof}

A normalization constant for the signed Laplacian can be found using the Gershgorin circle theorem. Let $A \in \CC^{n \times n}$ be a square matrix. Every eigenvalue $\lambda \in \CC$ of $A$ lies in at least one of the disks $C_i = \{ c \in \mathbb{C} \mid |c - A_{ii}| \leq r_i \}$ for $i = 1, \ldots, n$, where $r_i = \sum_{j=1, j \neq i}^n |A_{ij}|$.
For $A = \calL^{G_s}$, the signed Laplacian, the diagonal terms correspond to the degrees of the vertices and are real non-negative. The disks $C_i$ are contained in the larger disks $D_i = \{ c \in \mathbb{C} \mid |c| \leq r_i + A_{ii} \}$.
For all $i = 1, \ldots, n$, $A_{ii} = \deg(i) \leq S$, due to the upper bound on the degree (here $\deg(i)$ is the degree of the $i$-th vertex of the graph). Similarly, $r_i = \sum_{j \sim i} |s((i, j))| = \deg(i) \leq S$. It follows that any eigenvalue, which is real due to the Hermitianity of the Laplacian and non-negative due to its positive semi-definiteness, is upper-bounded by $2S$.

We define the classical circuit implementing $O_\text{row}^{\calL^{G_s}}$ as per the following algorithm:

\begin{algorithmic}[1]
\Statex \textbf{Function} $O_\text{row}(i, \ell)$
\Require $i \in 0, \ldots, 2^n-1$
\Require $\ell \in 0, \ldots, 2^m-1$
\State invalid $\gets (O_\text{adj}(i, 0) = 0$ and $O_\text{adj}(i, 1) \neq 0)$ or $i = 0$ \Comment{True iff $i \not\in V$}
\If{invalid and $\ell = 0$} 
    \State \textbf{return} $i$ \Comment{If the index $i$ does not correspond to any vertex in the graph, we set the entry $(i,i)$ to one: as such, we will avoid the vector $\ket{i}$ to contribute to the kernel of the overall matrix.}
\ElsIf{invalid and $\ell > 0$}
    \State \textbf{return} $0$ \Comment{If the index $i$ does not correspond to any vertex in the graph, every entry except the one lying on the diagonal has to be zero, as such we mark it with the placeholder location $0$.}
\Else
    \State \textbf{return} $O_\text{adj}(i, \ell)$
\EndIf
\end{algorithmic}

\medskip The role of $O_\text{row}$ is uniquely to marks integers that are not associated with a vertex. These will correspond, in $O_\text{entry}$, to zero entries for the row. We define the classical circuit implementing $O_\text{entry}^{\calL^{G_s}}$ as per the following algorithm:

\medskip \begin{algorithmic}[1]
\Statex \textbf{Function} $O_\text{entry}(a, i, \ell)$
\Require $a$ ancillary qubit
\Require $i \in 0, \ldots, 2^n-1$
\Require $\ell \in 0, \ldots, 2^m-1$
\State invalid $\gets (O_\text{adj}(i, 0) = 0$ and $O_\text{adj}(i, 1) \neq 0)$ or $i = 0$ \Comment{True iff $i \not\in V$}
\State $j \gets O_\text{row}(i, \ell)$
\If{invalid} \Comment{index not corresponding to any vertex}
    \If{$i = j$}
        \State Apply $R_y$ on qubit $a$ for an angle $1/\alpha$ \Comment{Contribution $\ketbra{i}$ for $i\not\in V$}
    \Else
        \State Apply no rotation on $a$
    \EndIf
\Else 
    \If{$i = j$}
        \State Calculate $\deg(i)$ \Comment{$\order{S}$ calls to $O_\text{adj}$}
        \State Apply $R_y$ on qubit $a$ for an angle $\deg(i)/\alpha$
    \ElsIf{$j \neq 0$}
        \State $s \gets O_\text{sign}(i, \ell)$
        \State Apply $R_y$ on qubit $a$ for an angle $-s/\alpha$
    \Else
        \State Apply no rotation
    \EndIf
\EndIf
\end{algorithmic}
\end{proof}

\restatabletheoremoneqma*
\begin{proof}
Consider a protocol in which Merlin provides an \( n \)-qubit witness state \( \ket{\psi} \), allegedly the ground state of the Hamiltonian \( H = \frac{1}{\alpha}(\calL^{G_s} \oplus \mathbb{I}_{\lnot V}) \), with $\alpha$ normalization constant that can be set to $2S$. We have been promised the ground state energy of $\calL^{G_s} \in \calH_V$ is either zero or $\delta$, while the ground state energy of the identity matrix $\mathbb{I}_{\lnot V} \in \calH_{\lnot V}$ is one. The ground state energy of $H \in \calH$ is the minimum between the ground state energies of the two contributions. We introduce the precision parameter $\delta' = \min \{\delta, 1\}/\alpha \in 1/\order{\poly n}$.
Arthur verifies the witness by applying QPE on the unitary \( H \), for which we can construct a block-encoding in time polynomial in \( n \), according to Propositions \ref{prop:ham_signed} and \ref{prop:efficient_be_laplacian}. The precision of the QPE is set to \( t = \lceil \log_2 (1/\delta') \rceil \) bits.

\medskip\noindent\emph{Completeness:} Let \( G_s \) be a \yes{}-instance of the \prbm{promise sparse balancedness} problem. 
For Proposition \ref{prop:balanced_kernel}, \( \ker \calL^{G_s} \neq \varnothing \). We also have \( \ker \calL^{G_s} = \ker H \). In this case, Merlin provides a witness \( \ket{\psi} \) such that \( \calL^{G_s} \ket{\psi} = 0 \). It follows that \( H \ket{\psi} = 0 \). The protocol gives an estimated energy of zero, and Arthur accepts the proof.

\medskip\noindent\emph{Soundness:} Let \( G_s \) be a \no{}-instance of the \prbm{promise sparse balancedness} problem. Then, for every witness \( \ket{\psi} \) that Merlin can provide, the energy satisfies $\expval{H}{\psi} 
\geq \frac{1}{\alpha} \min\{\expval{\calL^{G_s}}{\psi}, \expval{\mathbb{I}_{\lnot V})}{\psi}\} 
\geq \delta'$. The estimated energy is at least \( \delta' \). Therefore, Arthur rejects the proof.

\end{proof}

\section{QMA1-hardness of \prbm{sparse bipartitedness}}\label{sec:bipartite}

This section is dedicated to proving Theorem~\ref{thm:thm_2}.
In Section~\ref{sec:bipartite:connection}, we briefly recall the work of \textcite{zaslavsky2018negative} that connects balancedness to the bipartiteness of an unsigned graph, proposing a construction that maps balanced graphs to bipartite graphs. 
In Section~\ref{sec:bipartite:oracle}, we show that the construction proposed preserves the sparsity and efficiency of sparse access.
In Section \ref{sec:bipartite:model}, we show that our results extend to graphs expressed in the traditional sparse access.
In Section~\ref{sec:bipartite:reduction}, we demonstrate a reduction from \prbm{sparse balancedness} to the computational problem we define as \prbm{sparse bipartiteness}. Finally, in Section~\ref{sec:bipartite:containment} we prove that \prbm{sparse bipartiteness} under a suitable promise on the ground state energy of the signed Laplacian of the input graph is contained in \class{QMA}.

\subsection{Connection between balanced signed graphs and bipartite unsigned graphs}\label{sec:bipartite:connection}

We recall the concept of bipartiteness.

\begin{definition}{(Bipartite graph)}
An unsigned graph $G = (V, E)$ is \emph{bipartite} if there exists a partition of its vertices $V = A \cup B$, where $A \cap B = \varnothing$, such that each edge in $G$ connects a vertex in $A$ with a vertex in $B$. 
\end{definition}

\textcite{zaslavsky2018negative} proposed a construction that, given any signed graph, constructs a new unsigned graph of comparable size such that the balancedness of the former corresponds to the bipartiteness of the latter. This construction utilizes the subdivision operation, which removes an edge $\{v, v'\}$ from the graph and substitutes it with a new vertex $w$ and a pair of edges $\{v, w\}$ and $\{w, v'\}$. This work has already been used in~\cite{adriaens2023testing} to reduce the problem of testing balance to that of testing bipartiteness in the bounded degree model.

\begin{proposition}\label{propo:zas18}
Let $G_s = (V, E, s)$ be a signed graph, and denote by $E^+$ and $E^-$ the sets of positive and negative edges, respectively. Let $G_u = (V_u, E_u)$ be the unsigned graph obtained by applying the negative subdivision operation, i.e., replacing any positive edge with a path of two negative edges, and then ignoring the signature:
\begin{align*}
    V_u & = V \cup E^+, \\
    E_u & = E^- \cup \big\{ \{ v, e \} \mid v \in e \text{ and } e \in E^+ \big\}.
\end{align*}
Then, $G_s$ has a balanced component if and only if $G_u$ has a bipartite component.
\end{proposition}

\begin{proof}
See \textcite[Proposition 2.2]{zaslavsky2018negative}. 
\end{proof}

\subsection{Sparse access to the negative subdivision graph}\label{sec:bipartite:oracle}

We prove that the construction proposed in the previous subsection results in an unsigned graph that is represented efficiently. This is demonstrated in the following proposition.

\begin{proposition}\label{prop:sparse_access_bipartite}
Let $G_s = (V, E, s)$ be a signed graph with $N \leq 2^n - 1$ vertices, such that the graph is sparse, i.e., there is an upper bound $S \in \order{\poly n}$ on the degree of the vertices, and its sparse access $(O_\text{adj}, O_\text{sign})$ can be implemented using a classical circuit of size $\order{\poly n}$. Let $G_u = (V_u, E_u)$ be the graph obtained by applying the negative subdivision operation to each positively signed edge of $G_s$. Then, $G_u$ has $|V_u| = N' \leq N + NS$ vertices, is sparse with an upper bound $S' = \max{2, S}$ on the degree of its vertices, and its sparse access $O_\text{adj}'$ can be implemented using a classical circuit of size $\order{\poly n}$.
\end{proposition}
\begin{proof}
Let $S \le 2^m$. The set of vertices of $G_u$ is composed of the set $V$ plus as many vertices as the positively signed edges in $E$, therefore $V_u = V \cup E^+$; the cardinality of $V_u$ is $|V| + |E^+| \le N + NS \le (2^n - 1)(1 + 2^m) \le 2^{n+m+1}$. The upper bound $S'$ on the degree of the vertices is the maximum between the upper bound on the degree of $G_s$, which is $S$, and the maximum degree of the vertices associated with the negative subdivision operation, which is always $2$. The algorithm implementing $O_\text{adj}'$ is:

\bigskip
\begin{algorithmic}[1]
\Statex \textbf{Function} $O_\text{adj}'(i, \ell)$
\Require $i \in 0, \ldots, 2^{n+m+1}$
\Require $\ell \in 0, \ldots, \max\{2, S\} - 1$
\If{$0 \le i < 2^n$} \Comment{One of the vertices in $G_s$}
    \State invalid $\gets (O_\text{adj}(i, 0) = 0$ and $O_\text{adj}(i, 1) \neq 0)$ \Comment{True iff $i \not\in V$}
    \State $j \gets O_\text{adj}(i, \ell)$ 
    \If{invalid}
        \State \textbf{return} $\ell$
    \ElsIf{$j = 0$}
        \State \textbf{return} 0
    \ElsIf{$j \neq 0$ and $O_\text{sign}(i, \ell) = -1$} \Comment{Connect to the original endpoint of the edge} 
        \State \Return $j$
    \ElsIf{$j \neq 0$ and $O_\text{sign}(i, \ell) = +1$} \Comment{Connect to the vertex introduced by the negative subdivision}
        \State \Return $2^n + i \times 2^m + \ell$
    \EndIf
\EndIf
\If{$i \ge 2^n$} \Comment{One of the vertices introduced by the negative subdivision}
    \State $i', \ell' \gets \lfloor (i - 2^n)/2^m \rfloor, (i - 2^n) \text{ mod } 2^m$ 
    \State \Comment{From index $i$ possibly associated with a negative subdivision edge $e = \{i', j'\}$, retrieve $i'$} 
    \State \Comment{and $\ell'$ such that $j' = \operatorname{adj}(i', \ell')$. We must verify if $i$ is not associated to any such edge.}
    \State invalid $\gets (O_\text{adj}(i', 0) = 0$ and $O_\text{adj}(i', 1) \neq 0)$  \Comment{True iff $i \not\in V$}
    \If{invalid}
        \State \textbf{return} $\ell$
    \Else  
        \State $j' \gets O_\text{adj}(i', \ell')$
        \If{$j' = 0$}
            \State \textbf{return} $0$ \Comment{The index $i'$ is valid but has no $\ell$-th vertex}
        \ElsIf{$j' \neq 0$ and $O_\text{sign}(i', \ell') = -1$}
            \State \textbf{return} $0$ \Comment{The index $i'$ is valid, links to the vertex $j'$, but they are connected by a negative edge}
        \ElsIf{$j' \neq 0$ and $O_\text{sign}(i', \ell') = +1$}
            \State \textbf{return} $i'$ if $\ell = 0$ else ($j'$ if $\ell = 1$ else $0$)
        \EndIf
    \EndIf
\EndIf
\end{algorithmic}
\bigskip
\end{proof}

\subsection{Reduction}\label{sec:bipartite:reduction}

We define the computational problem of testing the bipartitedness of a sparse graph and prove that it is \qmaone{}-hard.

\restatablesparsebipartitedness*

\restatabletheoremtwo*
\begin{proof}
We will demonstrate hardness via a reduction from \prbm{sparse balancedness}.

\medskip\noindent\textbf{Reduction}:

\begin{enumerate}
    \item The input to \prbm{sparse balancedness} is the pair of classical circuits $(O_\text{adj}, O_\text{sign})$ implementing the sparse access to $G_s$.
    \item We use Proposition \ref{prop:sparse_access_bipartite} to obtain a classical circuit $O_\text{adj}'$ implementing the sparse access for the negative subdivision graph $G_u$. 
    \item We solve \prbm{sparse bipartitedness} with $O_\text{adj}'$ as input.
\end{enumerate}

\medskip\noindent\textbf{Completeness}: 

\begin{enumerate}
    \item Let $(O_\text{adj}, O_\text{sign})$ be a \yes{} instance of \prbm{sparse balancedness}. 
    \item For Theorem~\ref{prop:balanced_kernel}, $\ker\calL^{G_s} \neq \varnothing$.
    \item For Proposition~\ref{propo:zas18}, there exists an unsigned graph $G_u$ obtained by applying the negative subdivision operation such that $G_s$ is balanced if and only if $G_u$ is bipartite, and the classical circuit $O_\text{adj}'$ defined in the reduction above implement the sparse access for such $G_u$.
    \item $G_u$ is bipartite. Therefore, $O_\text{adj}'$ is a \yes{} instance of \prbm{sparse bipartitedness}.
\end{enumerate}

\medskip\noindent\textbf{Soundness}: 

\begin{enumerate}
    \item Let $(O_\text{adj}, O_\text{sign})$ be a \no{} instance of \prbm{sparse balancedness}. 
    \item For Theorem~\ref{prop:balanced_kernel}, $\ker\calL^{G_s} = \varnothing$.
    \item For Proposition~\ref{propo:zas18}, there exists an unsigned graph $G_u$ obtained by applying the negative subdivision operation such that $G_s$ is balanced if and only if $G_u$ is bipartite, and the classical circuit $O_\text{adj}'$ defined in the reduction above implement the sparse access for such $G_u$.
    \item $G_u$ is not bipartite. Therefore, $O_\text{adj}'$ is a \no{} instance of \prbm{sparse bipartitedness}.
\end{enumerate}
\end{proof}

\subsection{Role of the marked sparse access}\label{sec:bipartite:model}

The considerations detailed in Section \ref{sec:balance:model} regarding the role of marked sparse access for balance also apply to bipartiteness. For any unsigned graph \( G \), we can construct a signed graph \( G' \) with a comparable number of vertices and an upper bound on their degree such that \( G \) has a bipartite component if and only if \( G' \) has a bipartite component.

\begin{proposition}
Let $G = (V, E)$ be a unsigned graph with $V \subseteq [2^n] \setminus \{0\}, |V| = N$, such that
\begin{itemize}
    \setlength\itemsep{0em}
    \item the graph is sparse, i.e., there is an upper bound $2 \le S \in \order{\poly n}$ on the degree of the vertices;
    \item the graph is given as a classical circuit of size $\order{\poly n}$ implementing the marked sparse access $O_\text{adj}$.
\end{itemize} 
Then, there exists a unsigned graph $G_s' = (V', E')$ vertices, $V' = \{1, ..., N'\}$, such that
\begin{itemize}
    \setlength\itemsep{0em}
    \item $N' \le 2^{n+1}$;
    \item the upper bound on the degree of the vertices is $S + 2\in \order{\poly n}$;
    \item $G_s$ has a bipartite component iff $G_s'$ has a bipartite component. 
    \item it exists a circuit $O_\text{adj}', O_\text{sign}'$ implementing the traditional sparse access that uses $\order{1}$ calls to $O_\text{adj}$;
\end{itemize} 
\end{proposition}
\begin{proof}

The procedure is close to the one in Proposition \ref{prop:equiv_marked_access}. Let $A = \lceil 2^n/S \rceil$. Define $G_s' = (V', E', s')$ such that
\begin{align}
    V' & = \{1, ..., N'\} \text{ with } N' = 2^n + A + 3 \\
    E' & = E & (\text{all the edges in } G_s) \\
    & \phantom{=} \cup \{ \{i, i+1\} \mid i = 2^n, \ldots, 2^n+A \} & (\text{auxiliary vertices form a line}) \\
    & \phantom{=} \cup \{ \{ i, 2^n + \lceil i / S \rceil \} \mid i \not\in V \} & (\text{bitstring } i \not\in V \text{ connected to their nearest auxiliary vertex}) \\
    & \phantom{=} \cup \{ \{2^n+A+1, 2^n+A+2\}, \\
    & \quad \phantom{=} \{2^n+A+2, 2^n+A+3\}, \\
    & \quad \phantom{=} \{2^n+A+1, 2^n+A+3\} \} & (\text{triangle}).
\end{align}
By construction, \( G' \) has fewer than \( 2^{n+1} \) vertices for \( S \ge 2 \), and each auxiliary vertex has no more than \( S + 2 \) neighbors: two vertices preceding and succeeding it in the line, and \( S \) additional vertices corresponding to bitstrings not associated with vertices of \( G_s \). Denoting \( G^\text{new} \) as the subgraph composed solely of the augmented vertices (auxiliary ones and the triangle), the Laplacian of \( G' \) can be written as
\begin{equation}
    \calL^{G'} = \mqty[ \calL^{G} & \mathbf{0} \\ \mathbf{0} & \calL^{G^\text{new}} ].
\end{equation}
The subgraph \( G_s^\text{new} \) is not bipartite due to the triangle component, and as such, it cannot introduce new bipartite components. Being disconnected from any vertex originally in $G$, it cannot even modify the existing bipartite components. The traditional sparse access \( O_\text{adj}'  \) to $G'$ is implemented by the same algorithm in Proposition \ref{prop:equiv_marked_access} (only $O_\text{adj}'$).
\end{proof}

\subsection{Containment in \class{QMA}}\label{sec:bipartite:containment}

We prove that a promise variant of \prbm{sparse bipartitedness} is contained in \class{QMA}. As anticipated in the introduction, it is challenging to characterize the presence of bipartite components in terms of the graph Laplacian. It is easier to tackle the problem using a modification of such an operator, the \emph{signless Laplacian}. 

\begin{definition}\label{def:signless_laplacian}
The \emph{signless Laplacian} of a graph $G = (V, E)$ is a linear operator $\calQ^G: \calV^G \to \calV^G$ defined as 
\begin{equation}
    \mel{i}{\calQ^G}{j} = \begin{cases}
        \deg(i), & i = j \\
        1, & i \sim j \\
        0, & \text{otherwise}
    \end{cases}.
\end{equation}
\end{definition}

The signless Laplacian can be expressed as $\mathcal{Q} = \mathcal{D} + \mathcal{A}$, where $\mathcal{D}$ is the degree operator and $\mathcal{A}$ is the adjacency matrix. In contrast, the graph Laplacian is defined as $\mathcal{L} = \mathcal{D} - \mathcal{A}$. It is well-known in the literature that the kernel of the signless Laplacian provides insights into the connected bipartite components of the graph.

\begin{proposition}{(\textcite{desai1994characterization}, \textcite[Corollary 2.2]{cvetkovic2007signless})}\label{prop:signless_components}
Let $G = (V, E)$ be an unsigned graph and $\calQ^G$ be its signless Laplacian. The dimensionality of $\ker \calQ^G$ is equal to the number of bipartite components in $G$.
\end{proposition}

We recall the variant of \prbm{sparse bipartitedness} under a promise on the smallest eigenvalue of the signless Laplacian.

\restatablepromisesparsebipartitedness*

We can prove that such a problem is contained in \class{QMA}. The process is analogous to the one presented in Section \ref{sec:balance:containment}. While the presence of balanced components is characterized by the non-triviality of the kernel of the signed Laplacian, the presence of bipartite components is characterized by the non-triviality of the kernel of the signless Laplacian.

\begin{proposition}\label{prop:ham_unsigned}
Let $G = (V, E)$ be an unsigned graph with $N \leq 2^n - 1$ vertices. The Hamiltonian $H = \calQ^{G} + \sum_{i \not\in V} \ketbra{i}$ has dimension $\dim H = 2^n$ and kernel $\ker(H) = \ker(\calQ^{G})$.
\end{proposition}
\begin{proof}
The procedure is identical to that described in Proposition \ref{prop:ham_signed} with $G$ replacing $G_s$ and $\calQ^G$ replacing $\calL^{G_s}$. Notably, $\calQ^G$ can be thought of as a signed Laplacian for a signed graph $G$ whose edges are all negative. 
\end{proof}

\begin{proposition}\label{prop:efficient_be_laplacian_unsigned}
Let $G$ be an unsigned graph with $N \le 2^n - 1$ vertices such that
\begin{itemize}
    \setlength\itemsep{0em}
    \item the graph is sparse, i.e., there is an upper bound $S \in \order{\poly n}$ on the degree of the vertices, $S \le 2^m$;
    \item the graph is given as a classical circuit of size $\order{\poly n}$ implementing the sparse access $O_\text{adj}$.
\end{itemize} 
Then, there exists a quantum circuit realizing a $(2S, m, 0)$-block encoding of $\calQ^{G} + \sum_{i \not\in V} \ketbra{i}$ in time $\order{\poly n}$.
\end{proposition}
\begin{proof}
The procedure is identical to that described in Proposition \ref{prop:efficient_be_laplacian} with $G$ replacing $G_s$ and $\calQ$ replacing $\calL^{G_s}$. 
\end{proof}

\restatabletheoremtwoqma*
\begin{proof}
Consider a protocol in which Merlin provides an \( n \)-qubit witness state \( \ket{\psi} \), allegedly the ground state of the Hamiltonian \( H = \frac{1}{\alpha}(\calQ^{G} \oplus \mathbb{I}_{\lnot V}) \), with $\alpha$ normalization constant that can be set to $2S$. We have been promised the ground state energy of $\calQ^{G} \in \calH_V$ is either zero or $\delta$, while the ground state energy of the identity matrix $\mathbb{I}_{\lnot V} \in \calH_{\lnot V}$ is one. The ground state energy of $H \in \calH$ is the minimum between the ground state energies of the two contributions. We introduce the precision parameter $\delta' = \min \{\delta, 1\}/\alpha \in 1/\order{\poly n}$.
Arthur verifies the witness by applying QPE on the unitary \( H \), for which we can construct a block-encoding in time polynomial in \( n \), according to Propositions \ref{prop:ham_signed} and \ref{prop:efficient_be_laplacian}. The precision of the QPE is set to \( t = \lceil \log_2 (1/\delta') \rceil \) bits.

\medskip\noindent\emph{Completeness:} Let \( G \) be a \yes{}-instance of the \prbm{promise sparse bipartitedness} problem. 
For Proposition \ref{prop:balanced_kernel}, \( \ker \calQ^{G} \neq \varnothing \). We also have \( \ker \calQ^{G} = \ker H \). In this case, Merlin provides a witness \( \ket{\psi} \) such that \( \calQ^{G} \ket{\psi} = 0 \). It follows that \( H \ket{\psi} = 0 \). The protocol gives an estimated energy of zero, and Arthur accepts the proof.

\medskip\noindent\emph{Soundness:} Let \( G \) be a \no{}-instance of the \prbm{promise sparse bipartitedness} problem. Then, for every witness \( \ket{\psi} \) that Merlin can provide, the energy satisfies $\expval{H}{\psi} 
\geq \frac{1}{\alpha} \min\{\expval{\calQ^{G}}{\psi}, \expval{\mathbb{I}_{\lnot V})}{\psi}\} = \frac{1}{\alpha} \min \{ \delta, 1\} \geq \delta'$. The estimated energy is at least \( \delta' \). Therefore, Arthur rejects the proof.
\end{proof}

\section{Acknowledgements}

This publication is part of the project Divide \& Quantum (with project number 1389.20.241) of the research programme NWA-ORC which is (partly) financed by the Dutch Research Council (NWO). This work was also supported by the Dutch National Growth Fund (NGF), as part of the Quantum Delta NL programme. This work was also supported by the European Union’s Horizon Europe program through the ERC CoG BeMAIQuantum (Grant No. 101124342).
Part of the work for this manuscript was conducted while CG was affiliated with applied Quantum algorithms (aQa), Leiden University, the Netherlands.

\bibliography{biblio}

\end{document}